\documentclass[10pt,aps,prl,twocolumn]{revtex4-1}

\usepackage{mathrsfs}
\usepackage{bbm}
\usepackage{amsmath}
\usepackage{amssymb}
\usepackage{graphicx}
\usepackage{mathtools}

\makeatletter

\usepackage{hyperref}
\usepackage{color}
\definecolor{myblue}{RGB}{0,50,200}
\hypersetup{
colorlinks,
citecolor=myblue,
linkcolor=myblue,
urlcolor=myblue
}
\allowdisplaybreaks

\newcommand{\mca}{\mathcal}
\newcommand{\mbb}{\mathbb}

\newcommand{\msf}{\mathsf}
\newcommand{\mfr}{\mathfrak}
\newcommand{\mbm}[1]{\boldsymbol{#1}}
\newcommand{\avg}[1]{\langle #1\rangle}

\newcommand{\bra}[1]{\left( #1 \right)}
\newcommand{\bras}[1]{\left[ #1 \right]}
\newcommand{\brab}[1]{\left\{ #1 \right\}}
\newcommand{\pp}{\partial}
\newcommand{\Tr}[1]{{\rm tr}\left\{#1\right\}}

\newcommand{\Dr}[1]{\dot{#1}}
\newcommand{\peq}{p^{\rm eq}}

\begin{document}
\title{Geometrical Bounds of the Irreversibility in Markovian Systems}

\author{Tan Van Vu}
\email{tan@biom.t.u-tokyo.ac.jp}

\affiliation{Department of Information and Communication Engineering, Graduate
School of Information Science and Technology, The University of Tokyo,
Tokyo 113-8656, Japan}

\author{Yoshihiko Hasegawa}
\email{hasegawa@biom.t.u-tokyo.ac.jp}

\affiliation{Department of Information and Communication Engineering, Graduate
School of Information Science and Technology, The University of Tokyo,
Tokyo 113-8656, Japan}

\date{\today}

\begin{abstract}
We derive geometrical bounds on the irreversibility in both quantum and classical Markovian open systems that satisfy the detailed balance condition.
Using information geometry, we prove that irreversible entropy production is bounded from below by a modified Wasserstein distance between the initial and final states, thus strengthening the Clausius inequality in the reversible-Markov case.
The modified metric can be regarded as a discrete-state generalization of the Wasserstein metric, which has been used to bound dissipation in continuous-state Langevin systems.
Notably, the derived bounds can be interpreted as the quantum and classical speed limits, implying that the associated entropy production constrains the minimum time of transforming a system state.
We illustrate the results on several systems and show that a tighter bound than the Carnot bound for the efficiency of quantum heat engines can be obtained.
\end{abstract}

\maketitle

\emph{Introduction.}---
Irreversibility, which is quantified by entropy production, is a fundamental concept in classical and quantum thermodynamics \cite{Seifert.2012.RPP,Vinjanampathy.2016.CP,Deffner.2019}.
Most macroscopic natural phenomena are irreversible, although their microscopic physical processes are generally time-symmetric.
According to the second law of thermodynamics, a system undergoing an irreversible process generates (on average) a positive entropy amount $\Delta S_{\rm tot}\ge 0$.
This bound can be saturated only when operations are performed in the infinite-time quasistatic limit.
However, as real processes must be completed in finite time, they are accompanied by a certain dissipation.
Tightening the lower bound on entropy production not only deepens our understanding of how much heat must be dissipated, but also provides insights into quantum technologies such as quantum computation \cite{Verstraete.2009.NP} and quantum heat engines \cite{Kosloff.2014.ARPC}.

In recent years, many studies have characterized the dissipation of thermodynamic processes using information geometry \cite{Salamon.1983.PRL,Ruppeiner.1995.RMP,Feng.2008.PRL,Machta.2015.PRL,Rotskoff.2017.PRE.GeoApp,Ito.2018.PRL,Nicholson.2018.PRE,Large.2018.EPL,Scandi.2019.Q,Cafaro.2020.PRE,Bryant.2020.PNAS}, which is the application of techniques from differential geometry to the manifolds of probability distributions and density matrices \cite{Amari.2000}.
Reference \cite{Deffner.2010.PRL} showed that entropy production in a closed driven quantum system is bounded from below by the Bures length between the final state and the corresponding equilibrium state.
Following a similar approach, Ref. \cite{Mancino.2018.PRL} determined a geometrical upper bound for the equilibration processes of open quantum systems.
As is well known, in classical systems near equilibrium, irreversible entropy production is related to the distance between thermodynamic states \cite{Crooks.2007.PRL,Sivak.2012.PRL}.
Meanwhile, a lower bound on dissipation in terms of the Wasserstein distance \cite{Villani.2008} has been defined for nonequilibrium Markovian systems described by Langevin equations \cite{Aurell.2011.PRL,Aurell.2012.JSP,Dechant.2019.arxiv}.
Information geometry is useful for deriving other important relations, such as speed limits \cite{Pires.2016.PRX,Deffner.2017.NJP,Ito.2020.PRX,Funo.2019.NJP}, quantum work fluctuation-dissipation relation \cite{Miller.2019.PRL}, and the efficiency-power tradeoff in microscopic heat engines \cite{Brandner.2020.PRL}.

In this Letter, we enlarge the family of these universal relations by investigating quantum and classical open systems that satisfy the detailed balance condition.
These systems obey microscopically reversible Markovian dynamics \cite{Crooks.1998.JSP} and can be modeled as coupled to an infinite thermal reservoir.
Examples include equilibration processes, which have received considerable interest in nonequilibrium physics \cite{Goold.2016.JPA,Deutsch.2018.RPP,Ptaszynski.2019.PRL,Shiraishi.2019.PRL}.
Specifically, we derive geometrical lower bounds on the entropy production in reversible Markovian systems described by master equations.
The spaces of quantum states and discrete distributions are treated as Riemannian manifolds, on which the time evolution of a system state is described by a smooth curve.
By defining a modified Wasserstein metric, we prove that the entropy production is bounded from below by the square of the geodesic distance between the initial and final states divided by the process time [cf. Eqs.~\eqref{eq:QuaWasDisBou} and \eqref{eq:ClaWasDisBou}].
The derived bounds strengthen the Clausius inequality of the second law for reversible Markovian systems.
They can also be regarded as generalizations of the bounds reported in Refs.~\cite{Aurell.2011.PRL,Dechant.2019.arxiv} to the discrete-state quantum and classical systems.
The equality of these bounds is attained only when the system dynamics follow the shortest paths.
Our modified metric is a quantum generalization of the Wasserstein metric, which measures the distance between two distributions and is widely used in optimal transport problems \cite{Villani.2008}.
Interestingly, the obtained inequalities can be interpreted as speed limits \cite{Mandelstam.1945.JP,Margolus.1998.PD,Campo.2013.PRL,Taddei.2013.PRL,Deffner.2017.JPA,Okuyama.2018.PRL,Shiraishi.2018.PRL,Shanahan.2018.PRL}, which establish the trade-off relations between the speed and dissipation cost of a state transformation.
We numerically illustrate the results on a quantum Otto engine and a classical two-level system.

\emph{Riemannian geometry.}---
First, we briefly describe some relevant concepts of Riemannian geometry.
Let $M$ be a smooth Riemannian manifold equipped with a metric $g_p$ on the tangent space at each point $p\in M$.
Note that there is an infinite number of such metrics, as long as the linearity, symmetry, and positive-definite conditions are met.
Notably, there exists a family of monotone metrics that are contractive under physical maps \cite{Morozova.1991.JSM,Petz.1996.LAA,Hiai.2009.LAA}, a representative of which is the Fisher information metric \cite{Wootters.1981.PRD,Braunstein.1994.PRL}.
In the quantum case, $M$ can be the space of density operators $\rho$, which are positive (i.e., $\rho\ge 0$) and have unit trace (i.e., ${\rm tr}\,\rho=1$).
Meanwhile, in classical discrete-state systems, $M$ can be the collection of discrete distributions $\mbm{p}=[p_1,\dots,p_N]^\top$, where $p_n\ge 0$ and $\sum_{n=1}^{N}p_n=1$.
The length of a smooth curve $\{\gamma(t)\}_{0\le t\le \tau}$ on the manifold can be defined as $\ell(\gamma)\coloneqq\int_0^\tau\sqrt{g_{\gamma}(\dot{\gamma},\dot{\gamma})}dt$, where the dot denotes a time derivative.
The geodesic distance between two points can be then defined as the minimum length over all smooth curves $\gamma$ connecting those points.
Throughout this Letter, we use the standard notation $\avg{\cdot,\cdot}$ of the scalar inner product, i.e., $\avg{\mbm{x},\mbm{y}}=\mbm{x}^\top\mbm{y}$ for the classical case and $\avg{X,Y}=\Tr{X^\dagger Y}$ for the quantum case.

\emph{Bounds in Markovian quantum systems.}---
We first consider an open quantum system that is weakly coupled to a heat bath at the inverse temperature $\beta$.
The time evolution of the density operator $\rho(t)$ of this system is described by the Lindblad master equation \cite{Lindblad.1976.CMP,Gorini.1976.JMP}:
\begin{equation}
\Dr{\rho}=\mca{L}(\rho)\coloneqq -i[H(t),\rho]+\mca{D}(\rho),\label{eq:LinEqu}
\end{equation}
where $\mca{L}$ is the Lindblad operator, $H(t)$ is the Hamiltonian, and $\mca{D}(\rho)$ is the dissipator given by $\mca{D}(\rho)\coloneqq\sum_{\mu,\omega}\alpha_\mu(\omega)\bras{2L_{\mu}(\omega)\rho L_{\mu}^\dagger(\omega)-\brab{L_{\mu}^\dagger(\omega)L_{\mu}(\omega),\rho}}$.
Here, $\brab{\cdot,\cdot}$ is the anti-commutator and $L_\mu(\omega)$ is a jump operator that satisfies $L_{\mu}^\dagger(\omega)=L_{\mu}(-\omega)$ and $[L_{\mu}(\omega),H]=\omega L_{\mu}(\omega)$.
Note that jump operators and coupling coefficients can be time-dependent, but we omit the time notation for simplicity.
We also assume that the detailed balance condition $\alpha_\mu(\omega)=e^{\beta\omega}\alpha_\mu(-\omega)$ are satisfied and the system is ergodic \cite{Lidar.2019.arxiv} (i.e., $[L_\mu(\omega),X]=0$ for all $\mu,\omega$ if and only if $X$ is proportional to the identity operator).
These assumptions are sufficient conditions for the Gibbs state $\rho^{\rm eq}(t)\coloneqq e^{-\beta H(t)}/Z_\beta(t)$ to be the instantaneous stationary state of the Lindblad master equation, i.e., $\mca{L}[\rho^{\rm eq}(t)]=0$ \cite{Alicki.1987,Strasberg.2017.PRX}, where $Z_\beta(t)$ is the partition function.

The entropy growth of the open system during time period $\tau$ is $\Delta S_{\rm tot}=\int_0^\tau\sigma_{\rm tot}(t) dt$, where $\sigma_{\rm tot}(t)=\Dr{S}+\beta\Dr{Q}$ is the entropy production rate \cite{Breuer.2002}.
Here, $\Dr{S}=-\Tr{\Dr{\rho}(t)\ln\rho(t)}$ denotes the von Neumann entropy flux of the system and $\Dr{Q}=-\Tr{H(t)\Dr{\rho}(t)}$ denotes the heat flux dissipated from the system to the bath.
The entropy production rate can be rewritten as $\sigma_{\rm tot}(t)=-\avg{\ln\rho(t)-\ln\rho^{\rm eq}(t),\Dr{\rho}(t)}=-\frac{d}{dt}S(\rho(t)||\rho^{\rm eq}(t))$, where $S(\rho_1 ||\rho_2)\coloneqq\Tr{\rho_1(\ln\rho_1-\ln\rho_2)}$ is the relative entropy of $\rho_1$ with respect to $\rho_2$, and the time derivative does not act on $\rho^{\rm eq}(t)$.
$\sigma_{\rm tot}(t)$ is non-negative because the relative entropy is monotonic under completely-positive trace-preserving maps; thereby, one can obtain the Clausius inequality $\Delta S_{\rm tot}\ge 0$.

We now construct an operator $\mca{K}_\rho$, and alternatively express the Lindblad master equation [Eq.~\eqref{eq:LinEqu}] in the form $\Dr{\rho}=\mca{K}_\rho\bra{-\ln\rho+\ln\rho^{\rm eq}}$ \cite{Supp.PhysRev}.
For an arbitrary density operator $\rho$, we define a tilted operator $[\rho]_\theta(X)\coloneqq e^{-\theta/2}\int_0^1e^{s\theta}\rho^sX\rho^{1-s}ds$, where $\theta$ is a real number.
Using this operator, $\mca{K}_\rho$ can be explicitly constructed as $\mca{K}_\rho(\psi)\coloneqq i\beta^{-1}[\psi,\rho]+\mca{O}_\rho(\psi)$. Here, $\mca{O}_\rho(\psi)\coloneqq\sum_{\mu,\omega}e^{-\beta\omega/2}\alpha_\mu(\omega)[L_{\mu}(\omega),[\rho]_{\beta\omega}([L_{\mu}^\dagger(\omega),\psi])]$ is a self-adjoint positive operator, which can be interpreted as a quantum analog of the Onsager matrix.
For an arbitrary smooth curve $\{\gamma(t)\}_{0\le t\le\tau}$, there exists a unique vector field of traceless self-adjoint operators $\{\nu(t)\}_{0\le t\le\tau}$ such that $\Dr{\gamma}(t)=\mca{K}_\gamma[\nu(t)]$ for all $t$.
Exploiting this representation, one can define a metric $g$ under which the gradient flow of the instantaneous relative entropy equals the flow associated with the system dynamics \cite{Jordan.1998.JMA,Maas.2011.JFA,Carlen.2017.JFA,Rouze.2019.JMP}.
Specifically, we define the metric $g_\gamma(\Dr{\gamma},\Dr{\gamma})=\avg{\nu,\mca{K}_\gamma(\nu)}$, which is always non-negative because $\avg{\nu,\mca{K}_\gamma(\nu)}=\avg{\nu,\mca{O}_\gamma(\nu)}\ge 0$.
Although the operator $\nu(t)$ is implicitly obtained from $\Dr{\gamma}(t)$, it can be regarded as the generalized thermodynamic force, and $g_\gamma(\Dr{\gamma},\Dr{\gamma})$ is the quantum dissipation function \cite{Onsager.1931.PR}.
This can be clarified as considering the path generated by the system dynamics, i.e., $\Dr{\rho}=\mca{K}_\rho(\phi)$ and $g_\rho(\Dr{\rho},\Dr{\rho})=\sigma_{\rm tot}(t)$, where $\phi=-(\ln\rho-\ln\rho^{\rm eq})+c$ is a traceless self-adjoint operator.
In addition to the thermodynamic length $\ell(\gamma)$, the thermodynamic divergence of a path, defined as \cite{Crooks.2007.PRL}
\begin{equation}
\ell_{\rm q}(\gamma)^2\coloneqq\tau\int_0^\tau g_\gamma(\Dr{\gamma},\Dr{\gamma})dt,
\end{equation}
is a measure of the dissipation along the path.
Note that by the Cauchy--Schwarz inequality, $\ell_{\rm q}(\gamma)\ge\ell(\gamma)$.
A modified Wasserstein distance between two states $\rho_0$ and $\rho_\tau$ can be defined as $\mca{W}_{\rm q}(\rho_0,\rho_\tau)\coloneqq\inf_\gamma\{\ell_{\rm q}(\gamma)\}$, where the infimum is taken over smooth curves with end points $\rho_0$ and $\rho_\tau$.
For relaxation processes, $\mca{W}_{\rm q}$ is exactly the geodesic distance induced by the defined metric \footnote{Since $\ell_{\rm q}(\gamma)\ge\ell(\gamma)$, we have $\mca{W}_{\rm q}(\rho_0,\rho_\tau)\ge\inf_\gamma\{\ell(\gamma)\}$ in the general case. However, for relaxation processes (i.e., the operator $\mca{K}_\gamma$ is time-independent), it can be shown that $\mca{W}_{\rm q}(\rho_0,\rho_\tau)=\inf_\gamma\{\ell(\gamma)\}$, where the equality is attained with a constant-speed path \cite{Dolbeault.2008.CVPD}.}.
It has been shown that a clear-cut definition of the quantum Wasserstein distance, by the direct generalization of the classical one, is not achievable \cite{Agredo.2017.STO}. Our generalization here is based on the Benamou--Brenier flow formulation of the original $L^2$-Wasserstein \cite{Benamou.2000.NM,Carlen.2017.JFA,Rouze.2019.JMP}. Other generalized metrics based on quantum couplings \cite{Golse.2016.CMP,Agredo.2017.STO,Palma.2019.arxiv} and the Kantorovich--Rubinstein duality \cite{Chen.2017.CSL} have also been proposed in the literature.
From the definition of $\mca{W}_{\rm q}$, the first main result is a geometrical lower bound of the entropy production:
\begin{equation}
\Delta S_{\rm tot}\ge\frac{\mca{W}_{\rm q}(\rho(0),\rho(\tau))^2}{\tau}.\label{eq:QuaWasDisBou}
\end{equation}
Inequality \eqref{eq:QuaWasDisBou} indicates that the irreversible entropy production is lower bounded by the distance between the initial and final states.
This bound is stronger than the conventional second law of thermodynamics; it can also be interpreted as a quantum speed limit, as it limits the time required to transform the system state.
The limit is governed by dissipation and the geometrical distance between states.
To generalize the result to the infinite-dimensional Hilbert space, the existence and the construction of the operator $\nu(t)$ in the definition of the metric must be clarified.
Since the distance $\mca{W}_{\rm q}$ is usually difficult to compute explicitly, we provide a lower bound of $\mca{W}_{\rm q}$ in terms of the trace-like distance $\msf{d}_{\rm T}(\rho_0,\rho_\tau)=\sum_{n=1}^{N}|a_n-b_n|$, where $\{a_n\}$ and $\{b_n\}$ are increasing eigenvalues of $\rho_0$ and $\rho_\tau$, respectively.
Specifically, we prove that $\mca{W}_{\rm q}(\rho_0,\rho_\tau)^2\ge \msf{d}_{\rm T}(\rho_0,\rho_\tau)^2/4\mca{A}_{\rm T}$ \cite{Supp.PhysRev}, where $\mca{A}_{\rm T}\coloneqq\tau^{-1}\int_0^\tau\sum_{\mu,\omega}\alpha_\mu(\omega)\|L_{\mu}(\omega)\|_\infty^2 dt$ characterizes the time scale of the quantum system and $\|X\|_\infty$ denotes the spectral norm of the operator $X$.
Note that this lower bound on $\mca{W}_{\rm q}$ is not invariant under the well-known unitary transformation of jump operators, because the conditions of jump operators uniquely determine the parameterization of the dynamics.
Consequently, the entropy production is also bounded from below by the trace-like distance between the initial and final states, given by
\begin{equation}
\Delta S_{\rm tot}\ge\frac{\msf{d}_{\rm T}(\rho(0),\rho(\tau))^2}{4\tau\mca{A}_{\rm T}}.\label{eq:QuaTotVarBou}
\end{equation}

The Hamiltonian and jump operators of a system must be time-independent in order to equilibrate with the environment and reach a steady state.
Thus, during equilibration, the entropy production can be bounded by the distance $\msf{d}_{\rm E}(\rho_0,\rho_\tau)=|\Tr{H(\rho_0-\rho_\tau)}|$ of the average energy change \cite{Supp.PhysRev},
\begin{equation}
\Delta S_{\rm tot}\ge\frac{\msf{d}_{\rm E}(\rho(0),\rho(\tau))^2}{\tau\mca{A}_{\rm E}},\label{eq:QuaEneBou}
\end{equation}
where $\mca{A}_{\rm E}\coloneqq\sum_{\mu,\omega}\alpha_\mu(\omega)\omega^2\|L_{\mu}(\omega)\|_\infty^2$.
A tighter bound in terms of the square of the heat current to the reservoir \cite{Shiraishi.2016.PRL} and another bound in terms of the change in entropy of the system can also be obtained \cite{Supp.PhysRev}. However, these bounds are not tight in the zero-temperature limit, as compared to the bound reported in Ref.~\cite{Timpanaro.2020.PRL}.
Inequalities \eqref{eq:QuaTotVarBou} and \eqref{eq:QuaEneBou} provide lower bounds not only on the entropy production, but also on the equilibration time, which is an essential quantity in quantum-state preparation \cite{Short.2012.NJP}, and which aids our understanding of thermalization \cite{Goold.2016.JPA}.
In applications, the equilibration time can be approximated without solving the Lindblad master equation, which may be time-consuming in the weak coupling limit.
The dissipation-current trade-off relation \cite{Tajima.2020.arxiv}, which unveils the role of coherence between energy eigenstates in realizing a dissipation-less heat current, can also be derived using our geometrical approach \cite{Supp.PhysRev}.

The system becomes classical when the initial density matrix has no coherence in the energy eigenbasis of the Hamiltonian.
In what follows, we present the analysis for classical systems.

\emph{Bounds in Markovian classical systems.}---
Next, we consider a discrete-state system in contact with a heat bath at the inverse temperature $\beta$.
During a time period $\tau$, stochastic transitions between the states are induced by interactions with the heat bath.
The dynamics obey a time-continuous Markov jump process and are described by the master equation:
\begin{equation}
\Dr{p}_n(t)=\sum_{m(\neq n)}\bras{R_{nm}(t)p_m(t)-R_{mn}(t)p_n(t)},\label{eq:ClaMasEqu}
\end{equation}
where $p_n(t)$ is the probability of finding the system in state $n$ at time $t$, and $R_{mn}(t)$ is the (possibly time-dependent) transition rate from state $n$ to state $m$ ($1\le n\neq m\le N$).
We assume an irreducible system in which the transition rates satisfy the detailed balance condition $R_{nm}(t)e^{-\beta\mca{E}_m(t)}=R_{mn}(t)e^{-\beta\mca{E}_n(t)}$ for all $m\neq n$, where $\mca{E}_n(t)$ is the instantaneous energy of state $n$ at time $t$.
Herein, we define the instantaneous equilibrium state $\mbm{p}^{\rm eq}(t)$ as $\peq_n(t)\propto e^{-\beta\mca{E}_n(t)}$.

Within the stochastic thermodynamics framework \cite{Seifert.2012.RPP}, the irreversible entropy production $\Delta S_{\rm tot}$ is quantified by the change in the system's Shannon entropy and the heat flow dissipated into the environment.
Specifically, $\Delta S_{\rm tot}=\int_{0}^{\tau}\sigma_{\rm tot}(t)dt$, where $\sigma_{\rm tot}(t)=\sigma(t)+\sigma_{\rm M}(t)$ is the total entropy production rate.
The terms $\sigma(t)=\sum_{m,n}R_{mn}p_n\ln(p_n/p_m)$ and $\sigma_{\rm M}(t)=\sum_{m,n}R_{mn}p_n\ln(R_{mn}/R_{nm})$ define the entropy production rates of the system and medium, respectively.
Under the detailed balance condition, the entropy production rate can be explicitly calculated as $\sigma_{\rm tot}(t)=\avg{\mbm{f}(t),\Dr{\mbm{p}}(t)}=-\frac{d}{dt}D(\mbm{p}(t)||\mbm{p}^{\rm eq}(t))$, where $\mbm{f}(t)\coloneqq -\nabla_pD(\mbm{p}(t)||\mbm{p}^{\rm eq}(t))$ is a vector of thermodynamic forces, and the time derivative does not act on $\mbm{p}^{\rm eq}(t)$.
Here, $D(\mbm{p}||\mbm{q})=\sum_np_n\ln(p_n/q_n)$ is the relative entropy between the distributions $\mbm{p}$ and $\mbm{q}$, and $\nabla_p\coloneqq [\pp_{p_1},\dots,\pp_{p_N}]^\top$ denotes the gradient with respect to $\mbm{p}$.
The second law of thermodynamics, $\Delta S_{\rm tot}\ge 0$, is affirmed from the positivity of the entropy production rate $\sigma_{\rm tot}(t)$.

The master equation [Eq.~\eqref{eq:ClaMasEqu}] can be alternatively written as $\Dr{\mbm{p}}(t)=\msf{K}_p(t)\mbm{f}(t)$ \cite{Supp.PhysRev}, where $\msf{K}_p(t)$ is a symmetric positive semi-definite matrix, given by
\begin{equation}
\msf{K}_p(t)\coloneqq\sum_{n<m}R_{nm}(t)\peq_m(t)\Phi\bra{\frac{p_n(t)}{\peq_n(t)},\frac{p_m(t)}{\peq_m(t)}}\msf{E}_{nm}.
\end{equation}
Here, $\Phi(x,y)=(x-y)/[\ln(x)-\ln(y)]$ is the logarithmic mean of $x,y>0$ and $\msf{E}_{nm}=[e_{ij}]\in\mbb{R}^{N\times N}$ is a matrix with $e_{nn}=e_{mm}=1$, $e_{nm}=e_{mn}=-1$, and zeros in all other elements.
The symmetric matrix $\msf{K}_p$ is actually the Onsager matrix \cite{Onsager.1931.PR}, which linearly relates the thermodynamic forces to the probability currents.
For an arbitrary smooth curve $\{\mbm{\gamma}(t)\}_{0\le t\le\tau}$, there exists a unique vector field $\{\mbm{v}(t)\}_{0\le t\le\tau}$ such that $\Dr{\mbm{\gamma}}(t)=\msf{K}_\gamma(t)\mbm{v}(t)$ and $\avg{\mbm{1},\mbm{v}(t)}=0$, where $\mbm{1}\coloneqq [1,\dots,1]^\top$ is an all-ones vector.
We can thus define the Riemannian metric $g_\gamma(\Dr{\mbm{\gamma}},\Dr{\mbm{\gamma}})=\avg{\mbm{v},\msf{K}_\gamma\mbm{v}}$, which is always non-negative.
Using this metric, the thermodynamic divergence of a curve can be defined as
\begin{equation}
\ell_{\rm c}(\mbm{\gamma})^2\coloneqq\tau\int_0^\tau g_\gamma(\Dr{\mbm{\gamma}},\Dr{\mbm{\gamma}})dt.
\end{equation}
The modified Wasserstein distance between two points $\mbm{p}_0$ and $\mbm{p}_\tau$ is then defined as $\mca{W}_{\rm c}(\mbm{p}_0,\mbm{p}_\tau)\coloneqq\inf_\gamma\brab{\ell_{\rm c}(\mbm{\gamma})}$, where the infimum is taken over all smooth curves connecting $\mbm{p}_0$ and $\mbm{p}_\tau$ on the manifold.
Notably, this distance is bounded from below by the total variation distance \cite{Supp.PhysRev}.
It is worth noting that the defined metric is not equivalent to the traditional discrete version of the classical Wasserstein metric.
In practice, $\mca{W}_{\rm c}$ can be numerically calculated by the geodesic equation \cite{Supp.PhysRev}, which computes the shortest path between two points.
Defining $\mbm{h}(t)\coloneqq\mbm{f}(t)-N^{-1}\avg{\mbm{1},\mbm{f}(t)}\mbm{1}$, one observes that $\Dr{\mbm{p}}(t)=\msf{K}_p(t)\mbm{h}(t)$ and $\avg{\mbm{1},\mbm{h}(t)}=0$.
As $\sigma_{\rm tot}(t)=\avg{\mbm{h}(t),\msf{K}_p(t)\mbm{h}(t)}$, $\tau\Delta S_{\rm tot}$ is exactly the thermodynamic divergence of the path described by the system dynamics.
As the second main result, we obtain the following bound:
\begin{equation}
\Delta S_{\rm tot}\ge\frac{\mca{W}_{\rm c}(\mbm{p}(0),\mbm{p}(\tau))^2}{\tau}.\label{eq:ClaWasDisBou}
\end{equation}
Inequality \eqref{eq:ClaWasDisBou} provides a stronger bound than the Clausius inequality of the second law, and is valid as long as the transition rates satisfy the detailed balance condition.
Geometrically, Eq.~\eqref{eq:ClaWasDisBou} can be considered as a discrete-state generalization of the relation between dissipation and the Wasserstein distance, which has been studied in continuous-state Markovian dynamics governed by Langevin equations \cite{Aurell.2011.PRL,Dechant.2019.arxiv}.
Concretely, Eq.~(21) in Ref.~\cite{Aurell.2011.PRL} and Eq.~(2) in Ref.~\cite{Dechant.2019.arxiv} are referred to as the continuum analogs of Eq.~\eqref{eq:ClaWasDisBou}.
Our generalization newly and appropriately connects these thermodynamic and geometric quantities in the discrete case.
Therefore, it is applicable to the many discrete physical phenomena in biological and quantum physics.

\emph{Examples.}---
First, we illustrate the bounds derived in Eqs.~\eqref{eq:QuaTotVarBou} and \eqref{eq:QuaEneBou} on a quantum Otto heat engine \cite{Abah.2012.PRL,Kosloff.2017.E,Kloc.2019.PRE}, which consists of a two-level atom with the Hamiltonian $H(t)=\omega(t)\sigma_z/2$.
This system is alternatively coupled to two heat baths at different inverse temperatures [one hot, one cold, $\beta_k=1/T_k~(k=h,c)$], and is cyclically operated through four steps as demonstrated in Fig.~\ref{fig:result}(a).
During adiabatic expansion (compression), the isolated system unitarily evolves during time $\tau_a$, and its frequency changes from $\omega_h\to\omega_c$ ($\omega_c\to\omega_h$).
The dynamics in each isochoric process $k=h,c$ are described by the Lindblad master equation \cite{Breuer.2002}:
\begin{equation}
\begin{aligned}
\Dr{\rho}=-&i[H_k,\rho]+\alpha_k\bar{n}(\omega_k)(2\sigma_+\rho\sigma_--\{\sigma_-\sigma_+,\rho\})\\
+&\alpha_k(\bar{n}(\omega_k)+1)(2\sigma_-\rho\sigma_+-\{\sigma_+\sigma_-,\rho\}),
\end{aligned}
\end{equation}
where the frequency is fixed at $\omega_k$, $\sigma_{\pm}=(\sigma_x\pm i\sigma_y)/2$, $\alpha_k$ is a positive damping rate, and $\bar{n}(\omega_k)=(e^{\beta_k\omega_k}-1)^{-1}$ is the Planck distribution.
The density operator $\rho$ in this thermalization process is analytically solvable \cite{Ramezani.2018.PRE} and the total entropy production can be explicitly evaluated as $\Delta S_{\rm tot}^k=S(\rho(0)||\rho^{\rm eq})-S(\rho(\tau_k)||\rho^{\rm eq})$, where $\tau_k$ denotes the process time.
Equations~\eqref{eq:QuaTotVarBou} and \eqref{eq:QuaEneBou} constrain $\Delta S_{\rm tot}^k$ within the distances $\msf{d}_{\rm T}$ and $\msf{d}_{\rm E}$, as numerically verified in Fig.~\ref{fig:result}(b).
Note that unlike the classical case \cite{Shiraishi.2019.PRL}, $\Delta S_{\rm tot}^k$ in generic thermalization processes is not bounded by the relative entropy $S(\rho(0)||\rho(\tau_k))$ \cite{Supp.PhysRev}.
\begin{figure}[t]
\centering
\includegraphics[width=1.0\linewidth]{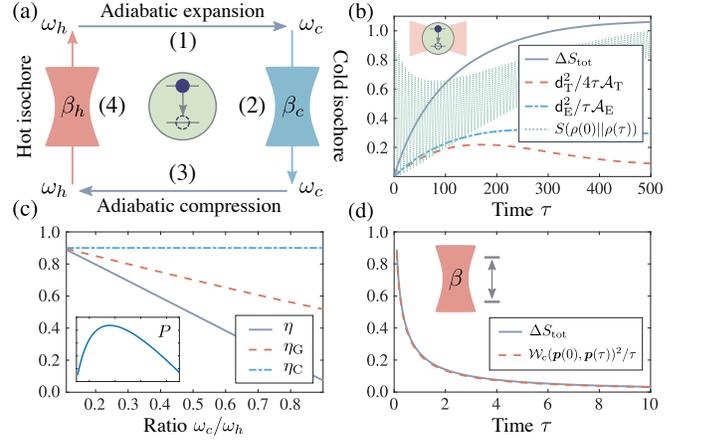}
\protect\caption{Numerical verification. (a) Quantum Otto engine: A two-level atom undergoes two isochoric and two adiabatic processes. (b) Thermalization process of the two-level atom. Plotted are $\Delta S_{\rm tot}$ (solid line), $\msf{d}_{\rm T}^2/4\tau\mca{A}_{\rm T}$ (dashed line), $\msf{d}_{\rm E}^2/\tau\mca{A}_{\rm E}$ (dash-dotted line), and $S(\rho(0)||\rho(\tau))$ (dotted line). Parameters are $\beta_k=1,\omega_k=1,\alpha_k=10^{-3}$, and $\rho(0)=(\mbb{I}_2+0.1\sigma_x-0.5\sigma_y+0.8\sigma_z)/2$. (c) Engine efficiency $\eta$ (solid line), Carnot efficiency $\eta_{\rm C}$ (dash-dotted line), and the derived efficiency bound $\eta_{\rm G}$ (dashed line), as functions of the cold-to-hot ratio of operating frequency. The inset plots the power output $P$ of the engine over the same frequency-ratio range. Parameters are $\beta_c=1,\beta_h=0.1,\alpha_h=\alpha_c=10^{-3}$, and $\tau_a=\tau_c=\tau_h=1$. (d) Classical two-level system. Plotted are $\Delta S_{\rm tot}$ (solid line) and $\mca{W}_{\rm c}(\mbm{p}(0),\mbm{p}(\tau))^2/\tau$ (dashed line). Parameters are fixed as $a=0.7,b=0.4$.}\label{fig:result}
\end{figure}

The total entropy production in each cycle is the sum of those in the hot and cold isochoric processes; that is, $\Delta S_{\rm tot}=\Delta S_{\rm tot}^h+\Delta S_{\rm tot}^c$.
Assuming a stationary-state system, let $Q_h$ and $Q_c$ denote the heat taken from the hot bath and the heat transferred to the cold bath, respectively.
From the inequality $\Delta S_{\rm tot}=\beta_hQ_h-\beta_cQ_c\ge 0$ imposed by the second law, one can prove that the engine efficiency cannot exceed the Carnot efficiency $\eta\coloneqq 1-{Q_c}/{Q_h}\le 1-{\beta_h}/{\beta_c}\eqqcolon\eta_{\rm C}$.
From the derived bounds, we can tighten the bound on the efficiency of the quantum Otto engine.
Applying Eqs.~\eqref{eq:QuaTotVarBou} and \eqref{eq:QuaEneBou} to isochoric processes, one readily obtains $\beta_hQ_h-\beta_cQ_c\ge\mfr{g}$, where $\mfr{g}\coloneqq\max\brab{\msf{d}_{\rm T}(\rho_1,\rho_4)^2/4\tau_h\mca{A}_{\rm T}^h,\msf{d}_{\rm E}(\rho_1,\rho_4)^2/\tau_h\mca{A}_{\rm E}^h}+\max\brab{\msf{d}_{\rm T}(\rho_2,\rho_3)^2/4\tau_c\mca{A}_{\rm T}^c,\msf{d}_{\rm E}(\rho_2,\rho_3)^2/\tau_c\mca{A}_{\rm E}^c}$.
Here, $\rho_i$ denotes the density operator at the beginning of process $i~(1\le i\le 4)$, $\mca{A}_{\rm T}^k\coloneqq\alpha_k(2\bar{n}(\omega_k)+1)$, and $\mca{A}_{\rm E}^k\coloneqq\omega_k^2\alpha_k(2\bar{n}(\omega_k)+1)$ for each $k=h,c$.
Consequently, the efficiency can be bounded from above as $\eta\le\eta_{\rm C}-{\mfr{g}}/{\beta_cQ_h}\eqqcolon\eta_{\rm G}$.
This bound is numerically verified in Fig.~\ref{fig:result}(c), which plots the efficiency against the $\omega_c/\omega_h$ ratio.

Next, we numerically verify the bound derived in Eq.~\eqref{eq:ClaWasDisBou} in a time-driven two-level classical system.
The instantaneous energies of states $1$ and $2$ are $\mca{E}_1(t)=\beta^{-1}\ln[(1-a+b(t+1)/\tau)/(a-bt/\tau)]$ and $\mca{E}_2(t)=0$, respectively, where $0<b<a<1$ are constants.
Their respective transition rates are $R_{12}(t)=1,R_{21}(t)=e^{\beta\mca{E}_1(t)}$.
The probability distribution and entropy production can be analytically calculated.
The entropy production and modified Wasserstein distance are plotted as functions of time $\tau$ in Fig.~\ref{fig:result}(d).
The entropy production at all times was tightly bounded from below by the distance $\mca{W}_{\rm c}$.
This result numerically verifies Eq.~\eqref{eq:ClaWasDisBou}.
As another example, the thermalization process of a three-level system is presented in Ref.~\cite{Supp.PhysRev}.

\emph{Conclusions.}---
In this Letter, we derived the geometrical bounds of irreversibility in both quantum and classical open systems, thus strengthening the Clausius inequality of the second law of thermodynamics.
Furthermore, the study results elucidate that, beyond the linear response regime, the entropy production can be geometrically characterized.
This finding sheds light on the problem of minimizing dissipation in discrete-state systems by methods of optimal control \cite{Aurell.2011.PRL}.
Interpreting the bounds as speed limits shows that the state-transformation speed is constrained by dissipation in quantum systems.
By investigating the information-geometric structure underlying the system dynamics, we lay the foundations for obtaining useful thermodynamic relations.
Exploring analogous bounds in generic systems, which violate the detailed balance condition, and for higher cumulants of dissipation \cite{Scandi.2020.PRR,Miller.2020.PRL}, would be promising research directions.

\begin{acknowledgments}
We thank Shin-ichi Sasa, Keiji Saito, Sosuke Ito, and Andreas Dechant for fruitful discussions.
We are also grateful to Keiji Saito for the careful reading of the manuscript and for the valuable comments.
This work was supported by Ministry of Education, Culture, Sports, Science and Technology (MEXT) KAKENHI Grant No. JP19K12153.
\end{acknowledgments}

\end{document}


\title{Supplemental Material for \\ ``Geometrical Bounds of the Irreversibility in Markovian Systems''}

\author{Tan Van Vu}

\affiliation{Department of Information and Communication Engineering, Graduate
	School of Information Science and Technology, The University of Tokyo,
	Tokyo 113-8656, Japan}

\author{Yoshihiko Hasegawa}

\affiliation{Department of Information and Communication Engineering, Graduate
	School of Information Science and Technology, The University of Tokyo,
	Tokyo 113-8656, Japan}

\pacs{}
\maketitle

This supplemental material describes the details of calculations introduced in the main text. 
The equations and figure numbers are prefixed with S [e.g., Eq.~(S1) or Fig.~S1]. 
Numbers without this prefix [e.g., Eq.~(1) or Fig.~1] refer to items in the main text.

\tableofcontents
\vspace{1cm}

\vspace{0.5cm}

Hereafter, we denote by $\mfr{L}(\mca{H})$ and $\mfr{H}(\mca{H})$ the sets of linear and self-adjoint operators, respectively, on a complex Hilbert space $\mca{H}$ with dimension $N>0$.
The inner product $\avg{\cdot,\cdot}$ is defined as $\avg{\mbm{x},\mbm{y}}=\mbm{x}^\top\mbm{y}$ for $\mbm{x},\mbm{y}\in\mbb{R}^{N\times 1}$ (classical case) and $\avg{X,Y}=\Tr{X^\dagger Y}$ for $X,Y\in\mfr{L}(\mca{H})$ (quantum case).

\section{Open quantum systems}
\subsection{Alternative expression of the Lindblad master equation}
Here we show that the Lindblad master equation can be written as
\begin{equation}
\Dr{\rho}=\mca{K}_\rho(-\ln\rho+\ln\rho^{\rm eq}),
\end{equation}
where $\mca{K}_\rho:\nu\mapsto i\beta^{-1}[\nu,\rho]+\mca{O}_\rho(\nu)$.
The operator $\mca{O}_\rho$ is defined by
\begin{equation}
\mca{O}_\rho(\nu):=\sum_{\mu,\omega}e^{-\beta\omega/2}\alpha_\mu(\omega)[L_{\mu}(\omega),[\rho]_{\beta\omega}([L_{\mu}^\dagger(\omega),\nu])].
\end{equation}
For any density operator $\rho=\sum_np_n\Mat{n}$, where $\sum_np_n=1$ and $\{\Kt{n}\}_n$ are orthonormal eigenvectors, we can express the tilted operator as
\begin{equation}
[\rho]_\theta(X)=e^{-\theta/2}\int_0^1e^{s\theta}\rho^sX\rho^{1-s}ds=\sum_{n,m}\Phi(e^{\theta/2}p_n,e^{-\theta/2}p_m)\Br{n}X\Kt{m}\Kt{n}\Br{m}.\label{eq:tilted.operator}
\end{equation}
Here, $\Phi(x,y)$ is the logarithmic mean of two positive numbers $x$ and $y$, given by $\Phi(x,y)=(x-y)/[\ln(x)-\ln(y)]$ for $x\neq y$ and $\Phi(x,x)=x$.
As $\rho^{\rm eq}=e^{-\beta H}/Z_\beta$ and $[\ln\rho,\rho]=0$, we have $i\beta^{-1}[-\ln\rho+\ln\rho^{\rm eq},\rho]=-i[H,\rho]$.
Thus, we need only show that
\begin{equation}
\mca{O}_\rho(-\ln\rho+\ln\rho^{\rm eq})=\sum_{\mu,\omega}\alpha_\mu(\omega)\bras{2L_\mu(\omega)\rho L_\mu^\dagger(\omega)-\brab{L_\mu^\dagger(\omega)L_\mu(\omega),\rho}}.\label{eq:operator.O}
\end{equation}
To this end, we first show that $[\rho]_{\theta}([X,\ln\rho]-\theta X)=e^{-\theta/2}X\rho-e^{\theta/2}\rho X$ for an arbitrary operator $X\in\mfr{L}(\mca{H})$ and $\theta\in\mbb{R}$.
This can be achieved through the following transformation:
\begin{subequations}
\begin{align}
[\rho]_{\theta}([X,\ln\rho]-\theta X)&=e^{-\theta/2}\int_0^1e^{\theta s}e^{s\ln\rho}(X\ln\rho-\ln\rho X-\theta X)e^{(1-s)\ln\rho}ds\\
&=-e^{-\theta/2}\int_0^1\bras{e^{\theta s}e^{s\ln\rho}(\ln\rho+\theta)Xe^{(1-s)\ln\rho}+e^{\theta s}e^{s\ln\rho}X(-\ln\rho)e^{(1-s)\ln\rho}}ds\\
&=-e^{-\theta/2}\int_0^1\frac{d}{ds}\bras{e^{(\ln\rho+\theta)s}Xe^{(1-s)\ln\rho}}ds\\
&=e^{-\theta/2}(Xe^{\ln\rho}-e^{\ln\rho+\theta}X)\\
&=e^{-\theta/2}X\rho-e^{\theta/2}\rho X.
\end{align}
\end{subequations}
Next, applying the relation $[L_{\mu}^\dagger(\omega),H]=-\omega L_{\mu}^\dagger(\omega)$, one immediately obtains
\begin{subequations}
\begin{align}
[\rho]_{\beta\omega}([L_{\mu}^\dagger(\omega),-\ln\rho+\ln\rho^{\rm eq}])&=[\rho]_{\beta\omega}([L_{\mu}^\dagger(\omega),-\ln\rho-\beta H])\\
&=-[\rho]_{\beta\omega}([L_{\mu}^\dagger(\omega),\ln\rho]+\beta[L_{\mu}^\dagger(\omega),H])\\
&=-[\rho]_{\beta\omega}([L_{\mu}^\dagger(\omega),\ln\rho]-\beta\omega L_{\mu}^\dagger(\omega))\\
&=e^{\beta\omega/2}\rho L_{\mu}^\dagger(\omega)-e^{-\beta\omega/2}L_{\mu}^\dagger(\omega)\rho.
\end{align}
\end{subequations}
Consequently, as $L_{\mu}^\dagger(\omega)=L_{\mu}(-\omega)$ and $\alpha_\mu(\omega)=e^{\beta\omega}\alpha_\mu(-\omega)$, one can verify Eq.~\eqref{eq:operator.O} as follows:
\begin{subequations}
\begin{align}
&\mca{O}_\rho(-\ln\rho+\ln\rho^{\rm eq})\\
&=\sum_{\mu,\omega}e^{-\beta\omega/2}\alpha_\mu(\omega)[L_{\mu}(\omega),[\rho]_{\beta\omega}([L_{\mu}^\dagger(\omega),-\ln\rho+\ln\rho^{\rm eq}])]\\
&=\sum_{\mu,\omega}e^{-\beta\omega/2}\alpha_\mu(\omega)[L_{\mu}(\omega),e^{\beta\omega/2}\rho L_{\mu}^\dagger(\omega)-e^{-\beta\omega/2}L_{\mu}^\dagger(\omega)\rho]\\
&=\sum_{\mu,\omega}\alpha_\mu(\omega)\bras{-e^{-\beta\omega}L_{\mu}(\omega)L_{\mu}^\dagger(\omega)\rho+L_{\mu}(\omega)\rho L_{\mu}^\dagger(\omega)+e^{-\beta\omega}L_{\mu}^\dagger(\omega)\rho L_{\mu}(\omega)-\rho L_{\mu}^\dagger(\omega)L_{\mu}(\omega)}\\
&=\sum_{\mu,\omega}\brab{\alpha_\mu(\omega)\bras{L_{\mu}(\omega)\rho L_{\mu}^\dagger(\omega)-\rho L_{\mu}^\dagger(\omega)L_{\mu}(\omega)}+\alpha_\mu(-\omega)\bras{L_{\mu}(-\omega)\rho L_{\mu}^\dagger(-\omega)-L_{\mu}^\dagger(-\omega)L_{\mu}(-\omega)\rho}}\\
&=\sum_{\mu,\omega}\alpha_\mu(\omega)\bras{2L_\mu(\omega)\rho L_\mu^\dagger(\omega)-\brab{L_\mu^\dagger(\omega)L_\mu(\omega),\rho}}.
\end{align}
\end{subequations}

\subsection{Properties of the quantum Wasserstein metric}
Here we provide several useful properties regarding the metric defined in the main text.
\begin{lemma}
The inner product $\avg{\cdot,\mca{O}_\rho(\cdot)}$ satisfies the conjugate-symmetry condition $\avg{\xi,\mca{O}_\rho(\nu)}=\avg{\nu,\mca{O}_\rho(\xi)}^*$ for all operators $\nu,\xi\in\mfr{L}(\mca{H})$.
Here, $*$ denotes the complex conjugate.
\end{lemma}
\begin{proof}
For an arbitrary operator $X\in\mfr{L}(\mca{H})$ and $\theta\in\mbb{R}$, we have
\begin{subequations}
\begin{align}
\avg{\xi,[X,[\rho]_{\theta}([X^\dagger,\nu])]}&=\Tr{\xi^\dagger[X,[\rho]_{\theta}([X^\dagger,\nu])]}\\
&=\Tr{[\xi^\dagger,X][\rho]_{\theta}([X^\dagger,\nu])}\\
&=\sum_{n,m}\Phi(e^{\theta/2}p_n,e^{-\theta/2}p_m)\Br{n}[X^\dagger,\nu]\Kt{m}\Br{m}[\xi^\dagger,X]\Kt{n},\label{eq:prop1}
\end{align}
\end{subequations}
where we have used Eq.~\eqref{eq:tilted.operator} in Eq.~\eqref{eq:prop1}.
Swapping $\xi$ and $\nu$, one obtains
\begin{subequations}
\begin{align}
\avg{\nu,[X,[\rho]_{\theta}([X^\dagger,\xi])]}^*&=\sum_{n,m}\Phi(e^{\theta/2}p_n,e^{-\theta/2}p_m)\Br{n}[X^\dagger,\xi]\Kt{m}^*\Br{m}[\nu^\dagger,X]\Kt{n}^*\\
&=\sum_{n,m}\Phi(e^{\theta/2}p_n,e^{-\theta/2}p_m)\Br{m}[\xi^\dagger,X]\Kt{n}\Br{n}[X^\dagger,\nu]\Kt{m}\\
&=\avg{\xi,[X,[\rho]_{\theta}([X^\dagger,\nu])]}.\label{eq:prop2}
\end{align}
\end{subequations}
As $\mca{O}_\rho(\nu)=\sum_{\mu,\omega}e^{-\beta\omega/2}\alpha_\mu(\omega)[L_{\mu}(\omega),[\rho]_{\beta\omega}([L_{\mu}^\dagger(\omega),\nu])]$, Eq.~\eqref{eq:prop2} implies that
\begin{equation}
\avg{\nu,\mca{O}_\rho(\xi)}^*=\avg{\xi,\mca{O}_\rho(\nu)}.\label{eq:consym1}
\end{equation}
\end{proof}

From Eq.~\eqref{eq:prop1}, one observes that
\begin{equation}
\avg{\xi,[X,[\rho]_{\theta}([X^\dagger,\xi])]}=\sum_{n,m}\Phi(e^{\theta/2}p_n,e^{-\theta/2}p_m)|\Br{n}[X^\dagger,\xi]\Kt{m}|^2\ge 0.
\end{equation}
Therefore, $\avg{\xi,\mca{O}_\rho(\xi)}\ge 0$ for an arbitrary operator $\xi$.
Equality is attained only when $[L_\mu^\dagger(\omega),\xi]=0$ for all $\mu$ and $\omega$.
When $\xi$ is a self-adjoint operator, i.e., $\xi^\dagger=\xi$, we have $\avg{\xi,\mca{K}_\rho(\xi)}=\avg{\xi,\mca{O}_\rho(\xi)}\ge 0$.

\begin{proposition}\label{proposition:vanish}
A self-adjoint operator $\nu$ satisfies $\mca{K}_\rho(\nu)=0$ if and only if $\nu$ is spanned by $\mbb{I}_N$.
\end{proposition}
\begin{proof}
As $\mca{K}_\rho(\mbb{I}_N)=0$, we need only show that if $\mca{K}_\rho(\nu)=0$, then $\nu$ is spanned by $\mbb{I}_N$.
Noting that $0=\avg{\nu,\mca{K}_\rho(\nu)}=\avg{\nu,\mca{O}_\rho(\nu)}$, we find that $\avg{\nu,\mca{O}_\rho(\nu)}=0$ only when $[L_\mu^\dagger(\omega),\nu]=0$ for all $\mu$ and $\omega$.
As the dynamics of the quantum system are ergodic, this implies that $\nu$ is spanned by $\mbb{I}_N$.

\end{proof}

\begin{proposition}\label{proposition:traceless}
$\mca{K}_\rho(\nu)$ is a traceless self-adjoint operator for all $\nu\in\mfr{H}(\mca{H})$.
\end{proposition}
\begin{proof}
The expression
\begin{equation}
\mca{K}_\rho(\nu)=i\beta^{-1}[\nu,\rho]+\mca{O}_\rho(\nu)=i\beta^{-1}[\nu,\rho]+\sum_{\mu,\omega}e^{-\beta\omega/2}\alpha_\mu(\omega)[L_{\mu}(\omega),[\rho]_{\beta\omega}([L_{\mu}^\dagger(\omega),\nu])]
\end{equation}
is a linear combination of commutators.
Therefore, $\Tr{\mca{K}_\rho(\nu)}=0$ is immediately derived.
Note that $(i\beta^{-1}[\nu,\rho])^\dagger=i\beta^{-1}[\nu,\rho]$, so we need only show that $\mca{O}_\rho(\nu)$ is self-adjoint.
Using the relations $[\rho]_\theta(X)^\dagger=[\rho]_{-\theta}(X^\dagger)$, $[X,Y]^\dagger=[Y^\dagger,X^\dagger]$, $e^{-\beta\omega/2}\alpha_\mu(\omega)=e^{\beta\omega/2}\alpha_\mu(-\omega)$, and $L_{\mu}^\dagger(\omega)=L_{\mu}(-\omega)$, we can prove that $\mca{O}_\rho(\nu)$ is self-adjoint as follows:
\begin{align}
\mca{O}_\rho(\nu)^\dagger&=\sum_{\mu,\omega}e^{-\beta\omega/2}\alpha_\mu(\omega)[L_{\mu}(\omega),[\rho]_{\beta\omega}([L_{\mu}^\dagger(\omega),\nu])]^\dagger\\
&=\sum_{\mu,\omega}e^{-\beta\omega/2}\alpha_\mu(\omega)[[\rho]_{\beta\omega}([L_{\mu}^\dagger(\omega),\nu])^\dagger,L_{\mu}(\omega)^\dagger]\\
&=\sum_{\mu,\omega}e^{\beta\omega/2}\alpha_\mu(-\omega)[[\rho]_{-\beta\omega}([\nu,L_{\mu}(\omega)]),L_{\mu}(-\omega)]\\
&=\sum_{\mu,\omega}e^{\beta\omega/2}\alpha_\mu(-\omega)[L_{\mu}(-\omega),[\rho]_{-\beta\omega}([L_{\mu}^\dagger(-\omega),\nu])]\\
&=\mca{O}_\rho(\nu).
\end{align}
\end{proof}

\begin{lemma}
For an arbitrary density operator $\rho$ and traceless self-adjoint operator $\vartheta$, there exists a unique traceless self-adjoint operator $\nu$ such that $\vartheta=\mca{K}_\rho(\nu)$.
\end{lemma}
\begin{proof}
Let $\mca{B}=\{\chi_{j,k}\}_{1\le j,k\le N}$ denote the set of generalized Gell-Mann matrices, which span the space of operators in the complex Hilbert space $\mca{H}$.
Specifically, $\chi_{j,k}$ can be expressed as follows:
\begin{equation}
\chi_{j,k}=\begin{cases}
E_{k,j}+E_{j,k}, & {\text{if}}~ j<k, \\
i(E_{k,j}-E_{j,k}), & {\text{if}}~ j>k, \\
\sqrt{\frac{2}{j(j+1)}}\bra{\sum_{l=1}^{j}E_{l,l}-jE_{j+1,j+1}}, & {\text{if}}~ j=k<N, \\
N^{-1}\mbb{I}_N, & {\text{if}}~ j=k=N, \\
\end{cases}
\end{equation}
where $E_{j,k}$ denotes a matrix with $1$ in the $jk$-th entry and $0$ elsewhere.
In this construction, each $\chi_{j,k}$ is a Hermitian matrix and $\Tr{\chi_{j,k}}=\delta_{jN}\delta_{kN}$ for all $(j,k)$.
For convenience, we define a set $\overline{\mca{B}}:=\mca{B}\setminus\{\chi_{N,N}\}$.
For arbitrary traceless self-adjoint operator $X$, there exists real coefficients $c_{j,k}\in\mbb{R}$ such that $X=\sum_{j,k}c_{j,k}\chi_{j,k}$.
Taking the trace of both sides of the equation, we obtain $0=\Tr{X}=\sum_{j,k}c_{j,k}\Tr{\chi_{j,k}}=c_{N,N}$.
This implies that $X$ can be expressed as a linear combination of matrices in $\overline{\mca{B}}$ with all real coefficients.

By propositions \ref{proposition:vanish} and \ref{proposition:traceless}, $\mca{K}_\rho(\chi_{j,k})$ is obviously a nonzero traceless self-adjoint operator for all $(j,k)\neq(N,N)$.
We now show that $\{\mca{K}_\rho(\chi_{j,k})\}_{(j,k)\neq(N,N)}$ is an independent set, i.e., $\sum_{(j,k)\neq(N,N)}c_{j,k}\mca{K}_\rho(\chi_{j,k})=0$ only when $c_{j,k}=0$ for all $j,k$.
Indeed, by the linearity of $\mca{K}_\rho$, we have
\begin{equation}
\sum_{(j,k)\neq(N,N)}c_{j,k}\mca{K}_\rho(\chi_{j,k})=\mca{K}_\rho\Big(\sum_{(j,k)\neq(N,N)}c_{j,k}\chi_{j,k}\Big)=0.
\end{equation}
Under proposition \ref{proposition:vanish}, $\sum_{(j,k)\neq(N,N)}c_{j,k}\chi_{j,k}$ must be spanned by $\mbb{I}_N$ ($=N\chi_{N,N}$), i.e., $\sum_{(j,k)\neq(N,N)}c_{j,k}\chi_{j,k}=-c_{N,N}\chi_{N,N}$ for some $c_{N,N}$.
This is equivalent to $\sum_{1\le j,k\le N}c_{j,k}\chi_{j,k}=0$.
As $\mca{B}$ is a basis of $\mca{H}$, this equivalence requires that $c_{j,k}=0$ for all $j,k$.

Because $\{\mca{K}_\rho(\chi_{j,k})\}_{(j,k)\neq(N,N)}$ has $N^2-1$ elements, we can add another matrix $\phi$ to form a new basis of $\mca{H}$.
In terms of the elements of the new basis, $\mbb{I}_N$ can then be expressed as
\begin{equation}
\mbb{I}_N=z\phi+\sum_{(j,k)\neq(N,N)}c_{j,k}\mca{K}_\rho(\chi_{j,k}),\label{eq:Id.new.basis}
\end{equation}
where $z$ is some complex number.
Taking the trace of both sides of Eq.~\eqref{eq:Id.new.basis}, we obtain $N=z\,\Tr{\phi}$, which indicates that $z\neq 0$.
Therefore, $\phi$ can be expressed in terms of $\mbb{I}_N$ and $\{\mca{K}_\rho(\chi_{j,k})\}_{(j,k)\neq(N,N)}$ as
\begin{equation}
\phi=z^{-1}\Big[\mbb{I}_N-\sum_{(j,k)\neq(N,N)}c_{j,k}\mca{K}_\rho(\chi_{j,k})\Big].\label{eq:phi.new.basis}
\end{equation}
Equation \eqref{eq:phi.new.basis} implies that an arbitrary matrix can be expressed as a linear combination of elements in the following set:
\begin{equation}
\mca{S}:=\{\mbb{I}_N\}\cup\{\mca{K}_\rho(\chi_{j,k})\}_{(j,k)\neq(N,N)}.
\end{equation}
Equivalently, $\mca{S}$ is a basis of $\mca{H}$.
Consequently, because $\mca{K}_\rho(\chi_{j,k})$ is traceless and self-adjoint, an arbitrary traceless self-adjoint operator $\vartheta$ can be expressed in terms of $\{\mca{K}_\rho(\chi_{j,k})\}_{(j,k)\neq(N,N)}$ with real coefficients $\{c_{j,k}\}$ as
\begin{equation}
\vartheta=\sum_{(j,k)\neq(N,N)}c_{j,k}\mca{K}_\rho(\chi_{j,k})=\mca{K}_\rho\Big(\sum_{(j,k)\neq(N,N)}c_{j,k}\chi_{j,k}\Big).
\end{equation}
Defining the traceless self-adjoint operator $\nu:=\sum_{(j,k)\neq(N,N)}c_{j,k}\chi_{j,k}$, one readily obtains $\vartheta=\mca{K}_\rho(\nu)$.
Finally, to prove the uniqueness of $\nu$, we assume two traceless self-adjoint operators $\nu_1$ and $\nu_2$ such that $\vartheta=\mca{K}_\rho(\nu_1)=\mca{K}_\rho(\nu_2)$, then $\mca{K}_\rho(\nu_1-\nu_2)=0$.
Applying the result in proposition \ref{proposition:vanish}, we have $\nu_1-\nu_2=z\mbb{I}_N$ for some $z\in\mbb{C}$.
Thus, $zN=\Tr{z\mbb{I}_N}=\Tr{\nu_1-\nu_2}=0\Rightarrow z=0$, which implies the uniqueness of $\nu$.

\end{proof}

\begin{lemma}
Given an arbitrary traceless self-adjoint operator $\nu$, the equality $\avg{\nu+\lambda\mbb{I}_N,\mca{K}_\rho(\nu+\lambda\mbb{I}_N)}=\avg{\nu,\mca{K}_\rho(\nu)}$ holds for an arbitrary number $\lambda\in\mbb{C}$.
\end{lemma}
\begin{proof}
Since $\mca{K}_\rho(\nu+\lambda\mbb{I}_N)=\mca{K}_\rho(\nu)+\mca{K}_\rho(\lambda\mbb{I}_N)=\mca{K}_\rho(\nu)$, we have
\begin{subequations}
\begin{align}
\avg{\nu+\lambda\mbb{I}_N,\mca{K}_\rho(\nu+\lambda\mbb{I}_N)}&=\avg{\nu+\lambda\mbb{I}_N,\mca{K}_\rho(\nu)}\\
&=\avg{\nu,\mca{K}_\rho(\nu)}+\avg{\lambda\mbb{I}_N,\mca{K}_\rho(\nu)}\\
&=\avg{\nu,\mca{K}_\rho(\nu)}+\lambda^*\Tr{\mca{K}_\rho(\nu)}\\
&=\avg{\nu,\mca{K}_\rho(\nu)},
\end{align}
\end{subequations}
where we have used the traceless property of $\mca{K}_\rho$ obtained in proposition \ref{proposition:traceless}.
\end{proof}

\subsection{Lower bound of the quantum Wasserstein distance in terms of the trace-like distance}
Here we derive the lower bound of the quantum Wasserstein distance $\mca{W}_{\rm q}(\rho_0,\rho_\tau)$ in terms of the trace-like distance.
From the definition of the quantum Wasserstein distance, given a fixed positive number $\delta>0$, there exists a smooth curve $\rho(t)$ with end points $\rho_0$ and $\rho_\tau$ such that
\begin{equation}
\tau\int_0^\tau\avg{\nu,\mca{K}_\rho(\nu)}dt\le\mca{W}_{\rm q}(\rho_0,\rho_\tau)^2+\delta.
\end{equation}
Here, $\nu(t)\in\mfr{H}(\mca{H})$ is a traceless self-adjoint operator that satisfies $\Dr{\rho}(t)=\mca{K}_\rho[\nu(t)]$.
Let $\rho(t)=\sum_np_n(t)\Mat{n(t)}$ be a spectral decomposition with an orthogonal basis $\langle n(t)|m(t)\rangle=\delta_{nm}$, and define the self-adjoint operator $\phi(t):=\sum_nc_n\Mat{n(t)}$, where $|c_n|\le 1$ are real constants to be determined later.
Evidently, $\phi(t)$ commutes with $\rho(t)$, i.e., $[\phi,\rho]=0$.
Now, using the relations $\Dr{\rho}=i\beta^{-1}[\nu,\rho]+\mca{O}_\rho(\nu)$ and $\avg{\phi,[\nu,\rho]}=0$, we have 
\begin{subequations}
\begin{align}
\sum_nc_n[p_n(\tau)-p_n(0)]&=\Tr{\int_0^\tau\phi(t)\dot{\rho}(t)dt}\\
&=\int_0^\tau\avg{\phi,i\beta^{-1}[\nu,\rho]+\mca{O}_\rho(\nu)}dt\\
&=\int_0^\tau\avg{\phi,\mca{O}_\rho(\nu)}dt\\
&\le\bra{\int_0^\tau\avg{\phi,\mca{O}_\rho(\phi)}dt}^{1/2}\bra{\int_0^\tau\avg{\nu,\mca{O}_\rho(\nu)}dt}^{1/2}\\
&\le\bra{\tau^{-1}\int_0^\tau\avg{\phi,\mca{O}_\rho(\phi)}dt}^{1/2}\bra{\mca{W}_{\rm q}(\rho_0,\rho_\tau)^2+\delta}^{1/2}.\label{eq:trace.bound.1}
\end{align}
\end{subequations}
The first term in the last inequality \eqref{eq:trace.bound.1} can be rewritten as
\begin{subequations}
\begin{align}
\avg{\phi,\mca{O}_\rho(\phi)}&=\sum_{\mu,\omega}e^{-\beta\omega/2}\alpha_\mu(\omega)\avg{\phi,[L_{\mu}(\omega),[\rho]_{\beta\omega}([L_{\mu}^\dagger(\omega),\phi])]}\\
&=\sum_{\mu,\omega}e^{-\beta\omega/2}\alpha_\mu(\omega)\Tr{[\phi,L_{\mu}(\omega)][\rho]_{\beta\omega}([L_{\mu}^\dagger(\omega),\phi])}\\
&=\sum_{\mu,\omega}e^{-\beta\omega/2}\alpha_\mu(\omega)\avg{[L_{\mu}^\dagger(\omega),\phi],[\rho]_{\beta\omega}([L_{\mu}^\dagger(\omega),\phi])}.\label{eq:puri.term.1}
\end{align}
\end{subequations}
Before proceeding, we prove the following result.
\begin{proposition}\label{prop1}
For an arbitrary operator $X$, a real number $\theta$, and density operator $\rho$, the inequality
\begin{equation}
\avg{X,[\rho]_\theta(X)}\le\frac{1}{2}(e^{\theta/2}+e^{-\theta/2})\|X\|_\infty^2
\end{equation}
holds, where $\|X\|_\infty$ denotes the spectral norm of the operator $X$.
\end{proposition}
\begin{proof}
Using Eq.~\eqref{eq:tilted.operator}, we have
\begin{equation}
\avg{X,[\rho]_\theta(X)}=\sum_{n,m}\Phi(e^{\theta/2}p_n,e^{-\theta/2}p_m)\Br{n}X\Kt{m}\Br{m}X^\dagger\Kt{n}.
\end{equation}
Applying the inequality $\Phi(x,y)\le(x+y)/2$ and the relation $\sum_{n}\Mat{n}=\mbb{I}_N$, we obtain
\begin{subequations}
\begin{align}
\avg{X,[\rho]_\theta(X)}&\le\frac{1}{2}\sum_{n,m}\bra{e^{\theta/2}p_n+e^{-\theta/2}p_m}\Br{n}X\Kt{m}\Br{m}X^\dagger\Kt{n}\\
&=\frac{1}{2}\sum_{n,m}e^{\theta/2}p_n\Br{n}X\Kt{m}\Br{m}X^\dagger\Kt{n}+\frac{1}{2}\sum_{m,n}e^{-\theta/2}p_m\Br{m}X^\dagger\Kt{n}\Br{n}X\Kt{m}\\
&=\frac{1}{2}\sum_{n}e^{\theta/2}p_n\Br{n}XX^\dagger\Kt{n}+\frac{1}{2}\sum_{m}e^{-\theta/2}p_m\Br{m}X^\dagger X\Kt{m}\label{eq:bound.tilted.op.c}\\
&\le\frac{1}{2}\sum_{n}e^{\theta/2}p_n\|X\|_\infty^2+\frac{1}{2}\sum_{m}e^{-\theta/2}p_m\|X\|_\infty^2\label{eq:bound.tilted.op.d}\\
&=\frac{1}{2}(e^{\theta/2}+e^{-\theta/2})\|X\|_\infty^2.\label{eq:bound.tilted.op.e}
\end{align}
\end{subequations}
Here we applied two facts: $\Br{n}XX^\dagger\Kt{n}\le\|X\|_\infty^2$ in Eq.~\eqref{eq:bound.tilted.op.d} and $\sum_np_n=1$ in Eq.~\eqref{eq:bound.tilted.op.e}.
\end{proof}

Returning to our problem, we apply proposition \ref{prop1} with $X=[L_{\mu}^\dagger(\omega),\phi]$ and $\theta=\beta\omega$, and hence obtain
\begin{equation}
\avg{[L_{\mu}^\dagger(\omega),\phi],[\rho]_{\beta\omega}([L_{\mu}^\dagger(\omega),\phi])}\le\frac{1}{2}(e^{-\beta\omega/2}+e^{\beta\omega/2})\|[L_{\mu}^\dagger(\omega),\phi]\|_\infty^2\le 2(e^{-\beta\omega/2}+e^{\beta\omega/2})\|L_{\mu}(\omega)\|_\infty^2.
\end{equation}
Here, we used the inequalities $\|[X,Y]\|_\infty\le\|XY\|_\infty+\|YX\|_\infty$ and $\|XY\|_\infty\le\|X\|_\infty\|Y\|_\infty$ for all $X,Y\in\mfr{L}(\mca{H})$.
Consequently, we have
\begin{equation}
\avg{\phi,\mca{O}_\rho(\phi)}\le 2\sum_{\mu,\omega}e^{-\beta\omega/2}\alpha_\mu(\omega)(e^{-\beta\omega/2}+e^{\beta\omega/2})\|L_{\mu}(\omega)\|_\infty^2=4\sum_{\mu,\omega}\alpha_\mu(\omega)\|L_{\mu}(\omega)\|_\infty^2.\label{eq:trace.bound.2}
\end{equation}
From Eqs.~\eqref{eq:trace.bound.1} and \eqref{eq:trace.bound.2}, we easily obtain the following inequality:
\begin{equation}
\mca{W}_{\rm q}(\rho_0,\rho_\tau)^2+\delta\ge\frac{(\sum_nc_n[p_n(\tau)-p_n(0)])^2}{4\tau^{-1}\int_0^\tau \sum_{\mu,\omega}\alpha_\mu(\omega)\|L_{\mu}(\omega)\|_\infty^2dt}.\label{eq:trace.bound.3}
\end{equation}
Setting $c_n={\rm sign}[p_n(\tau)-p_n(0)]$ and taking the limit $\delta\to 0$ in Eq.~\eqref{eq:trace.bound.3}, a lower bound of the quantum Wasserstein distance is obtained as
\begin{equation}
\mca{W}_{\rm q}(\rho_0,\rho_\tau)\ge\frac{\sum_n|p_n(\tau)-p_n(0)|}{2\sqrt{\tau^{-1}\int_0^\tau \sum_{\mu,\omega}\alpha_\mu(\omega)\|L_{\mu}(\omega)\|_\infty^2dt}}.\label{eq:trace.bound.4}
\end{equation}
From Eq.~\eqref{eq:trace.bound.4}, we wish to bound the Wasserstein distance by the trace-like distance $\msf{d}_{\rm T}(\rho_0,\rho_\tau)=\sum_{n=1}^{N}|a_n-b_n|$, where $a_1\le a_2\le\dots\le a_N$ and $b_1\le b_2\le\dots\le b_N$ are increasing eigenvalues of $\rho_0$ and $\rho_\tau$.
Given two arrays of real numbers, $\{x_n\}$ and $\{y_n\}$, one can prove that
\begin{equation}
\sum_n|x_n-y_n|\ge\sum_n|x_n-y_{\chi(n)}|,
\end{equation}
where $\{\chi(n)\}$ is a permutation of $\{n\}$ such that $y_{\chi(n)}\ge y_{\chi(m)}$ if $x_n\ge x_m$.
Therefore, $\sum_n|p_n(\tau)-p_n(0)|\ge \msf{d}_{\rm T}(\rho_0,\rho_\tau)$, so the bound in terms of the trace-like distance is written as
\begin{equation}
\mca{W}_{\rm q}(\rho_0,\rho_\tau)\ge\frac{\msf{d}_{\rm T}(\rho_0,\rho_\tau)}{2\sqrt{\tau^{-1}\int_0^\tau \sum_{\mu,\omega}\alpha_\mu(\omega)\|L_{\mu}(\omega)\|_\infty^2dt}}.
\end{equation}

\subsection{Lower bound of the entropy production in terms of the average energy-change distance}
Here we derive the lower bound of the entropy production $\Delta S_{\rm tot}$ in terms of the distance $\msf{d}_{\rm E}(\rho_0,\rho_\tau)=|\Tr{H(\rho_0-\rho_\tau)}|$.
The Lindblad master equation can be expressed as $\Dr{\rho}(t)=-i[H,\rho(t)]+\mca{O}_\rho[\phi(t)]$, where $\phi(t):=-\ln\rho(t)+\ln\rho^{\rm eq}$.
Using the relations $\Tr{H[H,\rho]}=0$ and $\Delta S_{\rm tot}=\int_0^\tau \avg{\phi,\mca{O}_\rho(\phi)}dt$, we obtain 
\begin{subequations}
\begin{align}
|\Tr{H(\rho_0-\rho_\tau)}|&=\left|\Tr{H\int_0^\tau\dot{\rho}(t)dt}\right|\\
&=\left|\int_0^\tau\avg{H,\mca{O}_\rho(\phi)}dt\right|\\
&\le\bra{\int_0^\tau\avg{H,\mca{O}_\rho(H)}dt}^{1/2}\bra{\int_0^\tau\avg{\phi,\mca{O}_\rho(\phi)}dt}^{1/2}\\
&=\bra{\int_0^\tau\avg{H,\mca{O}_\rho(H)}dt}^{1/2}\sqrt{\Delta S_{\rm tot}}.\label{eq:energy.bound.1}
\end{align}
\end{subequations}
The first term in Eq.~\eqref{eq:energy.bound.1} can be rewritten as
\begin{subequations}
\begin{align}
\avg{H,\mca{O}_\rho(H)}&=\sum_{\mu,\omega}e^{-\beta\omega/2}\alpha_\mu(\omega)\avg{H,[L_{\mu}(\omega),[\rho]_{\beta\omega}([L_{\mu}^\dagger(\omega),H])]}\\
&=\sum_{\mu,\omega}e^{-\beta\omega/2}\alpha_\mu(\omega)\Tr{[H,L_{\mu}(\omega)][\rho]_{\beta\omega}([L_{\mu}^\dagger(\omega),H])}\\
&=\sum_{\mu,\omega}e^{-\beta\omega/2}\alpha_\mu(\omega)\avg{[L_{\mu}^\dagger(\omega),H],[\rho]_{\beta\omega}([L_{\mu}^\dagger(\omega),H])}.\label{eq:term.1}
\end{align}
\end{subequations}
Applying proposition \ref{prop1} with $X=[L_\mu^\dagger(\omega),H]$ and $\theta=\beta\omega$, one obtains
\begin{equation}
\avg{[L_\mu^\dagger(\omega),H],[\rho]_{\beta\omega}([L_\mu^\dagger(\omega),H])}\le\frac{1}{2}(e^{-\beta\omega/2}+e^{\beta\omega/2})\|[L_\mu^\dagger(\omega),H]\|_\infty^2=\frac{1}{2}(e^{-\beta\omega/2}+e^{\beta\omega/2})\omega^2\|L_\mu(\omega)\|_\infty^2.
\end{equation}
Consequently, we have
\begin{equation}
\avg{H,\mca{O}_\rho(H)}\le \frac{1}{2}\sum_{\mu,\omega}e^{-\beta\omega/2}\alpha_\mu(\omega)(e^{-\beta\omega/2}+e^{\beta\omega/2})\omega^2\|L_{\mu}(\omega)\|_\infty^2=\sum_{\mu,\omega}\alpha_\mu(\omega)\omega^2\|L_{\mu}(\omega)\|_\infty^2.\label{eq:energy.bound.2}
\end{equation}
From Eqs.~\eqref{eq:energy.bound.1} and \eqref{eq:energy.bound.2}, we readily obtain the following inequality:
\begin{equation}
\Delta S_{\rm tot}\ge\frac{\msf{d}_{\rm E}(\rho_0,\rho_\tau)^2}{\tau\sum_{\mu,\omega}\alpha_\mu(\omega)\omega^2\|L_{\mu}(\omega)\|_\infty^2}.\label{eq:energy.bound.3}
\end{equation}

A tighter bound in terms of the square of the heat current can be analogously obtained.
Applying the Cauchy--Schwarz inequality, we have
\begin{subequations}
\begin{align}
\int_0^\tau |\dot{Q}|dt &=\int_0^\tau\left|\Tr{H\dot{\rho}(t)}\right|dt\\
&=\int_0^\tau|\avg{H,\mca{O}_\rho(\phi)}|dt\\
&\le\bra{\int_0^\tau\avg{H,\mca{O}_\rho(H)}dt}^{1/2}\bra{\int_0^\tau\avg{\phi,\mca{O}_\rho(\phi)}dt}^{1/2}\\
&=\bra{\int_0^\tau\avg{H,\mca{O}_\rho(H)}dt}^{1/2}\sqrt{\Delta S_{\rm tot}}.
\end{align}
\end{subequations}
Subsequently, applying Eq.~\eqref{eq:energy.bound.2}, we obtain
\begin{equation}
\Delta S_{\rm tot}\ge\frac{\bra{\int_0^\tau |\dot{Q}|dt}^2}{\tau\sum_{\mu,\omega}\alpha_\mu(\omega)\omega^2\|L_{\mu}(\omega)\|_\infty^2}.
\end{equation}
This bound is a quantum version of the classical bound reported in Ref.~\cite{Shiraishi.2016.PRL}, in which the total entropy production is bounded from below by the square of the heat current.

Analogously, we can also derive a lower bound on $\Delta S_{\rm tot}$ in terms of the change of entropy $\Delta S$ of the system.
First, we can bound $\Delta S$ from above as
\begin{subequations}
\begin{align}
|\Delta S|&\le\int_0^\tau |\Tr{\Dr{\rho}(t)\ln\rho(t)}|dt\\
&=\int_0^\tau|\avg{\ln\rho,\mca{O}_\rho(\phi)}|dt\\
&\le\bra{\int_0^\tau\avg{\ln\rho,\mca{O}_\rho(\ln\rho)}dt}^{1/2}\bra{\int_0^\tau\avg{\phi,\mca{O}_\rho(\phi)}dt}^{1/2}\\
&=\sqrt{\Upsilon}\sqrt{\Delta S_{\rm tot}},
\end{align}
\end{subequations}
where $\Upsilon :=\int_0^\tau\avg{\ln\rho,\mca{O}_\rho(\ln\rho)}dt$.
Then, the total entropy production can be bounded from below as
\begin{equation}
\Delta S_{\rm tot}\ge \frac{|\Delta S|^2}{\Upsilon}.
\end{equation}
Equivalently, the heat dissipated to the environment can be bounded from below as
\begin{equation}
\Delta Q\ge \beta^{-1}\bra{-\Delta S+\frac{|\Delta S|^2}{\Upsilon}}.\label{eq:heat.bound}
\end{equation}
This bound is relevant to the inequality reported in Ref.~\cite{Timpanaro.2020.PRL}, $\Delta Q\ge\mca{Q}(\mca{S}^{-1}(-\Delta S))$.
However, the bound in Eq.~\eqref{eq:heat.bound} is not tight in the zero-temperature limit.

\subsection{Current-dissipation trade-off}
Here we derive a trade-off relation between the heat current and dissipation, i.e., we derive an upper bound on the ratio $J^2/\sigma_{\rm tot}$, where $J:=\Tr{H(t)\Dr{\rho}(t)}$ is the heat flow from the heat bath to the system and $\sigma_{\rm tot}$ is the total entropy production rate, which characterizes the irreversibility.
Using the relations $\Dr{\rho}(t)=-i[H(t),\rho(t)]+\mca{O}_\rho[\phi(t)]$ and $\sigma_{\rm tot}=\avg{\phi(t),\mca{O}_\rho[\phi(t)]}$ and applying the Cauchy--Schwarz inequality, we obtain
\begin{subequations}
\begin{align}
J^2=|\Tr{H(t)\Dr{\rho}(t)}|^2&=|\avg{H(t),\mca{O}_\rho[\phi(t)]}|^2\\
&\le\avg{H(t),\mca{O}_\rho[H(t)]}\avg{\phi(t),\mca{O}_\rho[\phi(t)]}\\
&=\avg{H(t),\mca{O}_\rho[H(t)]}\sigma_{\rm tot}.
\end{align}
\end{subequations}
Note that
\begin{subequations}
\begin{align}
\avg{H(t),\mca{O}_\rho[H(t)]}&=\sum_{\mu,\omega}e^{-\beta\omega/2}\alpha_\mu(\omega)\avg{[L_{\mu}^\dagger(\omega),H(t)],[\rho]_{\beta\omega}([L_{\mu}^\dagger(\omega),H(t)])}	\\
&=\sum_{\mu,\omega}e^{-\beta\omega/2}\alpha_\mu(\omega)\omega^2\avg{L_{\mu}^\dagger(\omega),[\rho]_{\beta\omega}(L_{\mu}^\dagger(\omega))}.
\end{align}
\end{subequations}
From Eq.~\eqref{eq:bound.tilted.op.c}, one can prove that
\begin{subequations}
\begin{align}
\avg{H(t),\mca{O}_\rho[H(t)]}&\le\frac{1}{2}\sum_{\mu,\omega}e^{-\beta\omega/2}\alpha_\mu(\omega)\omega^2\bras{e^{\beta\omega/2}\Tr{L_{\mu}^\dagger(\omega)L_{\mu}(\omega)\rho}+e^{-\beta\omega/2}\Tr{L_{\mu}(\omega)L_{\mu}^\dagger(\omega)\rho}}\\
&=\sum_{\mu,\omega}\alpha_\mu(\omega)\omega^2\Tr{L_{\mu}^\dagger(\omega)L_{\mu}(\omega)\rho}\\
&=\Tr{\msf{L}\rho},
\end{align}
\end{subequations}
where $\msf{L}:=\sum_{\mu,\omega}\alpha_\mu(\omega)\omega^2L_{\mu}^\dagger(\omega)L_{\mu}(\omega)$.
Decomposing $\rho=\rho_{\rm bd}+\rho_{\rm nd}$, where
\begin{subequations}
\begin{align}
\rho_{\rm bd}&=\sum_e\Pi_{e}\rho\Pi_{e},\\
\rho_{\rm nd}&=\sum_{e\neq e'}\Pi_{e}\rho\Pi_{e'},
\end{align}
\end{subequations}
and $\Pi_e$ is the projection to the eigenspace of $H$ with eigenvalue $e$.
As $[L_{\mu}^\dagger(\omega)L_{\mu}(\omega),H(t)]=0$, $[L_{\mu}^\dagger(\omega)L_{\mu}(\omega),\Pi_e]=0$ for all $e$.
Therefore, the coherence between eigenstates with different energies vanishes in $\Tr{\msf{L}\rho}$, i.e., $\Tr{\msf{L}\rho}=\Tr{\msf{L}\rho_{\rm bd}}$.
The trade-off relation between the heat current and dissipation is thus obtained as
\begin{equation}
\frac{J^2}{\sigma_{\rm tot}}\le\avg{H,\mca{O}_\rho(H)}\le\Tr{\msf{L}\rho_{\rm bd}}.
\end{equation}
This inequality, known as the current-dissipation trade-off relation \cite{Tajima.2020.arxiv}, implies that the ratio $J^2/\sigma_{\rm tot}$ is not enhanced by coherence between eigenstates with different energies, but is enhanced by coherence between degenerate energy eigenstates.

\subsection{Invalidity of the bound in terms of the relative entropy}
Here we prove that the total entropy production in thermalization processes cannot be bounded from below by the relative entropy between the initial and final states.
In thermalization processes, the dynamics of the density operator are governed by the Lindblad equation
\begin{equation}
\Dr{\rho}=-i[H,\rho]+\sum_{\mu,\omega}\alpha_\mu(\omega)\bras{2L_{\mu}(\omega)\rho L_{\mu}^\dagger(\omega)-\brab{L_{\mu}^\dagger(\omega)L_{\mu}(\omega),\rho}}.
\end{equation}
The total entropy production can be explicitly expressed as $\Delta S_{\rm tot}=S(\rho(0)||\rho^{\rm eq})-S(\rho(\tau)||\rho^{\rm eq})$, where $S(\rho_1 ||\rho_2):=\Tr{\rho_1(\ln\rho_1-\ln\rho_2)}$ is the relative entropy of $\rho_1$ with respect to $\rho_2$.
If the relative entropy satisfies the reverse triangle inequality:
\begin{equation}
S(\rho(0)||\rho^{\rm eq})\ge S(\rho(0)||\rho(\tau))+S(\rho(\tau)||\rho^{\rm eq}),
\end{equation}
then $\Delta S_{\rm tot}\ge S(\rho(0)||\rho(\tau))$ and the dissipation can be further bound by the quantum Fisher information and Wigner--Yanase metrics \cite{Mancino.2018.PRL}.
However, this inequality holds in the classical case \cite{Shiraishi.2019.PRL} but not in the general quantum case.
As a simple counterexample, consider that $\alpha_\mu(\omega)\to 0$ for all $\mu$ and $\omega$.
In this vanishing coupling limit, the total entropy production vanishes because the relative entropy is invariant under a unitary transform.
On the other hand, $S(\rho(0)||\rho(\tau))$ is always positive because $\rho(t)$ is changed under the internal dynamics; thus $\Delta S_{\rm tot}<S(\rho(0)||\rho(\tau))$.

\subsection{Quantum Otto heat engine}
Consider a quantum Otto heat engine consisting of a two-level atom, whose energy levels (the excited state $\Kt{e}$ and the ground state $\Kt{g}$) are changed by an external controller.
The atom is alternatively coupled with two heat baths at different inverse temperatures $\beta_h<\beta_c$, and undergoes two isochoric and two adiabatic processes.
The system Hamiltonian is given by $H(t)=\omega(t)\sigma_z/2$, where $\sigma_z=\Mat{e}-\Mat{g}$ is the Pauli matrix in the $z$ direction.
The heat engine is cyclically operated as follows:
\begin{itemize}
	\item[1.] During adiabatic expansion in time $\tau_a$, the frequency changes from $\omega_h$ to $\omega_c$, and work is produced due to the change in internal energy. Here, the word {\it adiabatic} means that the system is isolated from the heat baths and there is no heat exchange during the process, although jumps between energy eigenstates may occur.
	\item[2.] During the cold isochore in time $\tau_c$, the atom is in contact with the cold bath and the frequency $\omega_c$ remains unchanged. In this process, heat $Q_c$ is transferred from the working medium to the cold bath.
	\item[3.] During adiabatic compression in time $\tau_a$, the frequency is reversed from $\omega_c$ to $\omega_h$ and work is done on the medium.
	\item[4.] During the hot isochore in time $\tau_h$, the atom is in contact with the hot bath and the frequency $\omega_h$ is fixed. In this process, heat $Q_h$ is extracted from the hot bath by the working medium.
\end{itemize}
During the adiabatic process, the dynamics of the density matrix are governed by the von Neumann equation
\begin{equation}
\Dr{\rho}(t)=-i[H(t),\rho(t)].
\end{equation}
During an isochoric process $k=h$ or $c$, the atom is thermalized and the time evolution of the density matrix follows the Lindblad master equation \cite{Breuer.2002}:
\begin{equation}
\Dr{\rho}(t)=-i[H_k,\rho(t)]+\mca{D}_k[\rho(t)],\label{eq:Lin.iso}
\end{equation}
where the dissipator $\mca{D}_k[\cdot]$ is defined by
\begin{equation}
\mca{D}_k[\rho]=\alpha_k\bar{n}(\omega_k)(2\sigma_+\rho\sigma_--\{\sigma_-\sigma_+,\rho\})+\alpha_k(\bar{n}(\omega_k)+1)(2\sigma_-\rho\sigma_+-\{\sigma_+\sigma_-,\rho\}).
\end{equation}
Here, $H_k=\omega_k\sigma_z/2$, $\sigma_{\pm}=(\sigma_x\pm i\sigma_y)/2$, $\alpha_k$ is a positive damping rate, and $\bar{n}(\omega_k)=(e^{\beta_k\omega_k}-1)^{-1}$ is the Planck distribution.

\subsubsection{Analytical solution of the density matrix}
In the stationary state, the density matrix can be analytically calculated as
\begin{equation}
\rho(t)=(e^{\lambda(t)}+e^{-\lambda(t)})^{-1}e^{\lambda(t)\sigma_z},\label{eq:den.mat.ans}
\end{equation}
where $\lambda(t)$ is a periodic function satisfying $\lambda(t+\tau)=\lambda(t)$ with $\tau=2\tau_a+\tau_h+\tau_c$.
In the following analysis, we determine the analytical form of $\lambda(t)$.
For any operators $X$ and $Y$, one can prove that
\begin{equation}
e^{-\lambda X}Ye^{\lambda X}=\sum_{n=0}^{\infty}\frac{(-\lambda)^n}{n!}[X,Y]_n,
\end{equation}
where the nested commutator is recursively defined as $[X,Y]_n=[X,[X,Y]_{n-1}]$ and $[X,Y]_0=Y$.
Using the relations $[\sigma_z,\sigma_{+}]=2\sigma_{+}$ and $[\sigma_z,\sigma_{-}]=-2\sigma_{-}$, one readily obtains
\begin{subequations}
\begin{align}
[\sigma_z,\sigma_{+}]_n&=(+2)^n\sigma_{+},\\
[\sigma_z,\sigma_{-}]_n&=(-2)^n\sigma_{-}.
\end{align}
\end{subequations}
Subsequently,
\begin{subequations}
\begin{align}
e^{-\lambda\sigma_z}\sigma_{+}e^{\lambda\sigma_z}&=\sum_{n=0}^{\infty}\frac{(-\lambda)^n(2)^n}{n!}\sigma_{+}=e^{-2\lambda}\sigma_{+}\Leftrightarrow e^{\lambda\sigma_z}\sigma_{+}=e^{2\lambda}\sigma_{+}e^{\lambda\sigma_z},\\
e^{-\lambda\sigma_z}\sigma_{-}e^{\lambda\sigma_z}&=\sum_{n=0}^{\infty}\frac{(-\lambda)^n(-2)^n}{n!}\sigma_{-}=e^{2\lambda}\sigma_{-}\Leftrightarrow e^{\lambda\sigma_z}\sigma_{-}=e^{-2\lambda}\sigma_{-}e^{\lambda\sigma_z}.
\end{align}
\end{subequations}
Noting that $\sigma_+\sigma_-=(\mbb{I}_2+\sigma_z)/2$ and $\sigma_-\sigma_+=(\mbb{I}_2-\sigma_z)/2$, the dissipator term can be calculated as
\begin{subequations}\label{eq:dis.term}
\begin{align}
\mca{D}_k[e^{\lambda\sigma_z}]&=\alpha_k\bar{n}(\omega_k)(2\sigma_+e^{\lambda\sigma_z}\sigma_--\{\sigma_-\sigma_+,e^{\lambda\sigma_z}\})+\alpha_k(\bar{n}(\omega_k)+1)(2\sigma_-e^{\lambda\sigma_z}\sigma_+-\{\sigma_+\sigma_-,e^{\lambda\sigma_z}\})\\
&=2\bras{\alpha_k\bar{n}(\omega_k)(e^{-2\lambda}\sigma_+\sigma_--\sigma_-\sigma_+)+\alpha_k(\bar{n}(\omega_k)+1)(e^{2\lambda}\sigma_-\sigma_+-\sigma_+\sigma_-)}e^{\lambda\sigma_z}\\
&=\alpha_k\bras{\bar{n}(\omega_k)(e^{-2\lambda}-1)+(\bar{n}(\omega_k)+1)(e^{2\lambda}-1)+\brab{\bar{n}(\omega_k)(e^{-2\lambda}+1)-(\bar{n}(\omega_k)+1)(e^{2\lambda}+1)}\sigma_z}e^{\lambda\sigma_z}
\end{align}
\end{subequations}
Inserting Eqs.~\eqref{eq:den.mat.ans} and \eqref{eq:dis.term} into Eq.~\eqref{eq:Lin.iso}, we obtain
\begin{equation}
\bras{-\frac{e^{\lambda}-e^{-\lambda}}{e^{\lambda}+e^{-\lambda}}+\sigma_z}\Dr{\lambda}(t)=\alpha_k\bras{\bar{n}(\omega_k)(e^{-2\lambda}-1)+(\bar{n}(\omega_k)+1)(e^{2\lambda}-1)+\brab{\bar{n}(\omega_k)(e^{-2\lambda}+1)-(\bar{n}(\omega_k)+1)(e^{2\lambda}+1)}\sigma_z},
\end{equation}
which is satisfied if $\lambda(t)$ obeys the following differential equation:
\begin{equation}
\Dr{\lambda}(t)=\alpha_k\bras{\bar{n}(\omega_k)(e^{-2\lambda}+1)-(\bar{n}(\omega_k)+1)(e^{2\lambda}+1)}.
\end{equation}
This equation can be analytically solved for $\lambda(t)$.
The result is
\begin{equation}
\lambda(t)=\begin{cases}
\lambda(\tau_a), & {\rm if}~0\le t<\tau_a,\\
\frac{1}{2}\ln\bras{\frac{\bar{n}(\omega_c)e^{2\alpha_c(2\bar{n}(\omega_c)+1)(t-\tau_a)}-z_c}{(\bar{n}(\omega_c)+1)e^{2\alpha_c(2\bar{n}(\omega_c)+1)(t-\tau_a)}+z_c}}, & {\rm if}~\tau_a\le t<\tau_a+\tau_c,\\
\lambda(2\tau_a+\tau_c), & {\rm if}~\tau_a+\tau_c\le t<2\tau_a+\tau_c,\\
\frac{1}{2}\ln\bras{\frac{\bar{n}(\omega_h)e^{2\alpha_h(2\bar{n}(\omega_h)+1)(t-2\tau_a-\tau_c)}-z_h}{(\bar{n}(\omega_h)+1)e^{2\alpha_h(2\bar{n}(\omega_h)+1)(t-2\tau_a-\tau_c)}+z_h}}, & {\rm if}~2\tau_a+\tau_c\le t< \tau,
\end{cases}
\end{equation}
where the constant $z_k$ can be explicitly determined through the boundary conditions as
\begin{subequations}
\begin{align}
z_c&=[2\bar{n}(\omega_c)+1]\bra{\frac{\bar{n}(\omega_c)}{2\bar{n}(\omega_c)+1}-\frac{\bar{n}(\omega_h)}{2\bar{n}(\omega_h)+1}}\Bigg/\bras{1+\frac{1-e^{-2\alpha_c(2\bar{n}(\omega_c)+1)\tau_c}}{e^{2\alpha_h(2\bar{n}(\omega_h)+1)\tau_h}-1}},\\
z_h&=[2\bar{n}(\omega_h)+1]\bra{\frac{\bar{n}(\omega_h)}{2\bar{n}(\omega_h)+1}-\frac{\bar{n}(\omega_c)}{2\bar{n}(\omega_c)+1}}\Bigg/\bras{1+\frac{1-e^{-2\alpha_h(2\bar{n}(\omega_h)+1)\tau_h}}{e^{2\alpha_c(2\bar{n}(\omega_c)+1)\tau_c}-1}}.
\end{align}
\end{subequations}

\subsubsection{Thermodynamics and efficiency}
For each $1\le i\le 4$, let $\rho_i$ denote the density matrix at the beginning of process $i$.
Note that the density matrix is unchanged during the adiabatic processes, i.e., $\rho_1=\rho_2$ and $\rho_3=\rho_4$.
During an isochoric process, a heat quantity $Q_c=\Tr{H_c(\rho_2-\rho_3)}$ is transferred to the cold bath, or a heat quantity $Q_h=\Tr{H_h(\rho_1-\rho_4)}$ is extracted from the hot bath.
The total work $W$ extracted from the working medium is
\begin{equation}
-W=\int_0^{\tau_a}\Tr{\pp_tH(t)\rho(t)}dt+\int_{\tau_a+\tau_c}^{\tau-\tau_h}\Tr{\pp_tH(t)\rho(t)}dt=\Tr{\rho_1(H_c-H_h)}+\Tr{\rho_3(H_h-H_c)}.
\end{equation}
By conservation of energy, we have $-W+Q_h-Q_c=0$.
The efficiency $\eta$ is then defined as
\begin{equation}
\eta:=\frac{W}{Q_h}=1-\frac{Q_c}{Q_h}.
\end{equation}
The total entropy produced during the isochoric processes is
\begin{subequations}
\begin{align}
\Delta S_{\rm tot}^{h}&=\Delta S_{h}-\beta_hQ_h\ge 0,\\
\Delta S_{\rm tot}^{c}&=\Delta S_{c}+\beta_cQ_c\ge 0,
\end{align}
\end{subequations}
where $\Delta S_{h}=\Tr{\rho_4\ln\rho_4}-\Tr{\rho_1\ln\rho_1}$ and $\Delta S_{c}=\Tr{\rho_2\ln\rho_2}-\Tr{\rho_3\ln\rho_3}$ are the changes in the von Neumann entropy during the hot and cold isochoric processes, respectively.
As $\Delta S_{h}+\Delta S_{c}=0$, $\beta_cQ_c-\beta_hQ_h\ge 0$ follows from the second law of thermodynamics.
Using this inequality, one can prove that $\eta$ cannot exceed the Carnot efficiency, given by
\begin{equation}
\eta\le 1-\frac{\beta_h}{\beta_c}=:\eta_{\rm C}.
\end{equation}
In the following, we tighten the bound on the efficiency.
According to Eqs.~\TraceBoundQua~ and \EnerBoundQua~ in the main text, the total entropy productions during the isochoric processes are bounded from below by the distances $\msf{d}_{\rm T}(\cdot,\cdot)$ and $\msf{d}_{\rm E}(\cdot,\cdot)$ as follows:
\begin{align}
\Delta S_{\rm tot}^{h}&=\Delta S_{h}-\beta_hQ_h\ge\max\brab{\frac{\msf{d}_{\rm T}(\rho_1,\rho_4)^2}{4\tau_h\mca{A}_{\rm T}^h},\frac{\msf{d}_{\rm E}(\rho_1,\rho_4)^2}{\tau_h\mca{A}_{\rm E}^h}},\label{eq:bound.hot}\\
\Delta S_{\rm tot}^{c}&=\Delta S_{c}+\beta_cQ_c\ge\max\brab{\frac{\msf{d}_{\rm T}(\rho_2,\rho_3)^2}{4\tau_c\mca{A}_{\rm T}^c},\frac{\msf{d}_{\rm E}(\rho_2,\rho_3)^2}{\tau_c\mca{A}_{\rm E}^c}}.\label{eq:bound.cold}
\end{align}
Here, $\mca{A}_{\rm T}^k:=\alpha_k(2\bar{n}(\omega_k)+1)$ and $\mca{A}_{\rm E}^k:=\omega_k^2\alpha_k(2\bar{n}(\omega_k)+1)$ for $k=h$ or $c$.
From Eqs.~\eqref{eq:bound.hot} and \eqref{eq:bound.cold}, we obtain
\begin{equation}
\beta_cQ_c-\beta_hQ_h\ge\max\brab{\frac{\msf{d}_{\rm T}(\rho_1,\rho_4)^2}{4\tau_h\mca{A}_{\rm T}^h},\frac{\msf{d}_{\rm E}(\rho_1,\rho_4)^2}{\tau_h\mca{A}_{\rm E}^h}}+\max\brab{\frac{\msf{d}_{\rm T}(\rho_2,\rho_3)^2}{4\tau_c\mca{A}_{\rm T}^c},\frac{\msf{d}_{\rm E}(\rho_2,\rho_3)^2}{\tau_c\mca{A}_{\rm E}^c}}=:\mfr{g}.
\end{equation}
Consequently, a tighter bound on $\eta$ is obtained as
\begin{equation}
\eta\le\eta_{\rm C}-\frac{\mfr{g}}{\beta_cQ_h}=:\eta_{\rm G}.
\end{equation}

\section{Classical Markov jump processes}
\subsection{Alternative expression of the classical master equation}

We now show that the master equation $\Dr{\mbm{p}}=\msf{R}\mbm{p}$ can be expressed as $\Dr{\mbm{p}}=\msf{K}_p\mbm{f}$, where $\msf{R}=[R_{mn}]$ with $R_{nn}=-\sum_{m(\neq n)}R_{mn}$, $\msf{K}_p=\sum_{1\le n<m\le N}R_{nm}\peq_m\Phi\bra{\frac{p_n}{\peq_n},\frac{p_m}{\peq_m}}\msf{E}_{nm}$, and $\mbm{f}=-\nabla_{\mbm{p}} D(\mbm{p}||\mbm{p}^{\rm eq})$.
Here, $\nabla_{\mbm{p}}:=[\pp_{p_1},\dots,\pp_{p_N}]^\top$.
Specifically, we need to show that
\begin{equation}
(\msf{K}_p\mbm{f})_n=\sum_{m(\neq n)}[R_{nm}p_m-R_{mn}p_n]\label{eq:mastereq.row}
\end{equation}
holds for all $n$.
Indeed, using the relations $f_n=-(\ln p_n-\ln\peq_n - 1)$ and $R_{nm}\peq_m=R_{mn}\peq_n$, Eq.~\eqref{eq:mastereq.row} can be verified as follows:
\begin{subequations}
\begin{align}
(\msf{K}_p\mbm{f})_n&=\sum_{m(\neq n)}R_{nm}\peq_m\Phi\bra{\frac{p_n}{\peq_n},\frac{p_m}{\peq_m}}(\msf{E}_{nm}\mbm{f})_n\\
&=\sum_{m(\neq n)}R_{nm}\peq_m\Phi\bra{\frac{p_n}{\peq_n},\frac{p_m}{\peq_m}}(f_n-f_m)\\
&=\sum_{m(\neq n)}R_{nm}\peq_m\frac{p_n/\peq_n-p_m/\peq_m}{\ln p_n-\ln\peq_n-\ln p_m+\ln\peq_m}(\ln p_m-\ln\peq_m-\ln p_n+\ln\peq_n)\\
&=\sum_{m(\neq n)}R_{nm}\peq_m\bra{\frac{p_m}{\peq_m}-\frac{p_n}{\peq_n}}\\
&=\sum_{m(\neq n)}[R_{nm}p_m-R_{mn}p_n].
\end{align}
\end{subequations}

\subsection{Properties of the matrix $\msf{K}_p$}
The matrix $\msf{K}_p$ is symmetric and positive semi-definite.
Its properties are given below.
\begin{lemma}
For an arbitrary distribution $\mbm{p}$ satisfying $p_n>0$ for all $n$, ${\rm ker}(\msf{K}_p)=\{\mbm{v}\in\mbb{R}^{N\times 1}~|~\mbm{v}\propto\mbm{1}\}$.
\end{lemma}
\begin{proof}
As the system is irreducible, there exists a set of $N-1$ unordered pairs, $\mca{Y}=\{(i,j)\,|\,R_{ij}\neq 0\}$, such that for arbitrary states $n\neq m$, there is a path $n=i_0\to i_1\to\dots\to i_k=m$ and $(i_{l},i_{l+1})\in\mca{Y}$ for all $0\le l\le k-1$.
Assuming $\mbm{v}\in{\rm ker}(\msf{K}_p)$, we have
\begin{equation}
0=\avg{\mbm{v},\msf{K}_p\mbm{v}}=\sum_{1\le n<m\le N}R_{nm}\peq_m\Phi\bra{\frac{p_n}{\peq_n},\frac{p_m}{\peq_m}}(v_m-v_n)^2.
\end{equation}
This expression means that $v_i-v_j=0$ for all $(i,j)\in\mca{Y}$, or equivalently, $\mbm{v}\propto\mbm{1}$.
\end{proof}

\begin{lemma}
There exists a vector $\mbm{v}$ for which $\Dr{\mbm{p}}=\msf{K}_p\mbm{v}$.
Such a vector is unique under the condition $\avg{\mbm{1},\mbm{v}}=0$.
\end{lemma}
\begin{proof}
For any $\mbm{v}$ satisfying $\msf{K}_p\mbm{v}=0$ [i.e., $\mbm{v}\in {\rm ker}(\msf{K}_p)$], then $\mbm{v}\propto\mbm{1}\Rightarrow\mbm{v}^\top\Dr{\mbm{p}}=0$; equivalently, $\Dr{\mbm{p}}\in{\rm ker}(\msf{K}_p)^\perp$.
According to the Fredholm alternative, the equation $\Dr{\mbm{p}}=\msf{K}_p\mbm{v}$ always has a nonzero solution $\mbm{v}$.
Defining $\overline{\mbm{v}}:=\mbm{v}-N^{-1}\avg{\mbm{1},\mbm{v}}\mbm{1}$, we can write $\Dr{\mbm{p}}=\msf{K}_p\overline{\mbm{v}}$ and $\avg{\mbm{1},\overline{\mbm{v}}}=0$.
Assume that there exist two solutions $\mbm{v}_1$ and $\mbm{v}_2$ satisfying $\avg{\mbm{1},\mbm{v}_1}=\avg{\mbm{1},\mbm{v}_2}=0$.
We then have $\msf{K}_p(\mbm{v}_1-\mbm{v}_2)=0\Rightarrow\mbm{v}_1-\mbm{v}_2=c\mbm{1}$ for some $c\in\mbb{R}$.
Moreover, $\avg{\mbm{1},\mbm{v}_1-\mbm{v}_2}=0\Rightarrow Nc=0\Rightarrow c=0$, which proves the uniqueness of $\mbm{v}$.
\end{proof}

\subsection{Geodesic equation of the modified Wasserstein distance}
We here derive the geodesic equation that determines the shortest path between two distributions $\mbm{p}_0$ and $\mbm{p}_\tau$.
To this end, we consider the following functional, which is minimized along the geodesic path $\{\mbm{p}(t)\}_{0\le t\le\tau}$:
\begin{equation}
\mca{J}[\mbm{p}(t)]=\int_0^\tau\avg{\mbm{v}(t),\msf{K}_p\mbm{v}(t)}dt,
\end{equation}
where $\mbm{v}(t)$ and $\mbm{p}(t)$ are related through $\Dr{\mbm{p}}(t)=\msf{K}_p\mbm{v}(t)$.
Consider an arbitrary perturbation path $\{\mbm{q}(t)\}_{0\le t\le \tau}$ that satisfies $\mbm{q}(0)=\mbm{q}(\tau)=0$ and $\sum_nq_n(t)=0$ for all $0\le t\le \tau$.
Because the functional $\mca{J}[\mbm{\gamma}(t)]$ is minimized when $\mbm{\gamma}=\mbm{p}$, the function $\Theta(\epsilon)=\mca{J}[\mbm{p}(t)+\epsilon \mbm{q}(t)]$ has a minimum at $\epsilon=0$, so $\Theta'(0)=0$.
The functional evaluated at $\mbm{\gamma}=\mbm{p}+\epsilon\mbm{q}$ can be written as
\begin{align}
\mca{J}[\mbm{p}(t)+\epsilon \mbm{q}(t)]=\int_0^\tau\avg{\mbm{\vartheta}(t),\msf{K}_{p+\epsilon q}\mbm{\vartheta}(t)}dt,\label{eq:functional}
\end{align}
where $\mbm{\vartheta}(t)$ is determined from $\Dr{\mbm{p}}(t)+\epsilon\Dr{\mbm{q}}(t)=\msf{K}_{p+\epsilon q}\mbm{\vartheta}(t)$.
From Eq.~\eqref{eq:functional}, we have
\begin{equation}
0=\Theta'(0)=\int_0^\tau\bras{\avg{\pp_\epsilon\mbm{\vartheta}(t),\msf{K}_p\mbm{v}(t)}+\avg{\mbm{v}(t),\pp_\epsilon\msf{K}_{p+\epsilon q}\mbm{v}(t)}+\avg{\mbm{v}(t),\msf{K}_p\pp_\epsilon\mbm{\vartheta}(t)}}_{\epsilon=0}dt.\label{eq:pardiff.eps}
\end{equation}
Hereafter, we omit the notation of evaluating at $\epsilon=0$ for conciseness.
The first and third terms in Eq.~\eqref{eq:pardiff.eps} are equal by symmetry of $\msf{K}_p$; that is, $\avg{\pp_\epsilon\mbm{\vartheta}(t),\msf{K}_p\mbm{v}(t)}=\avg{\mbm{v}(t),\msf{K}_p\pp_\epsilon\mbm{\vartheta}(t)}$.
Taking the partial derivative of both sides of $\Dr{\mbm{p}}(t)+\epsilon\Dr{\mbm{q}}(t)=\msf{K}_{p+\epsilon q}\mbm{\vartheta}(t)$ with respect to $\epsilon$ and evaluating at $\epsilon=0$, we obtain
\begin{equation}
\Dr{\mbm{q}}(t)=\pp_\epsilon\msf{K}_{p+\epsilon q}\mbm{v}(t)+\msf{K}_p\pp_\epsilon\mbm{\vartheta}(t)\Rightarrow\avg{\mbm{v}(t),\msf{K}_p\pp_\epsilon\mbm{\vartheta}(t)}=\avg{\mbm{v}(t),\Dr{\mbm{q}}(t)}-\avg{\mbm{v}(t),\pp_\epsilon\msf{K}_{p+\epsilon q}\mbm{v}(t)}.\label{eq:pardiff.master}
\end{equation}
From Eqs.~\eqref{eq:pardiff.eps} and \eqref{eq:pardiff.master}, we have
\begin{equation}
0=\int_0^\tau\bras{2\avg{\mbm{v}(t),\Dr{\mbm{q}}(t)}-\avg{\mbm{v}(t),\pp_\epsilon\msf{K}_{p+\epsilon q}\mbm{v}(t)}}dt=-\int_0^\tau\bras{2\avg{\Dr{\mbm{v}}(t),\mbm{q}(t)}+\avg{\mbm{v}(t),\pp_\epsilon\msf{K}_{p+\epsilon q}\mbm{v}(t)}}dt.\label{eq:tmp1}
\end{equation}
Since
\begin{equation}
\pp_\epsilon\msf{K}_{p+\epsilon q}=\sum_{1\le n<m\le N}R_{nm}\peq_m\pp_\epsilon\Phi\bra{\frac{p_n+\epsilon q_n}{\peq_n},\frac{p_m+\epsilon q_m}{\peq_m}}\msf{E}_{nm},
\end{equation}
we have
\begin{equation}
\avg{\mbm{v}(t),\pp_\epsilon\msf{K}_{p+\epsilon q}\mbm{v}(t)}=\sum_{m,n}R_{mn}[v_m(t)-v_n(t)]^2\Psi\bra{\frac{p_n(t)}{\peq_n(t)},\frac{p_m(t)}{\peq_m(t)}}q_n(t)=\avg{\mbm{r}(t),\mbm{q}(t)},\label{eq:tmp2}
\end{equation}
where $\Psi(x,y)=[x-\Phi(x,y)]/[x(\ln x-\ln y)]$ and $r_n(t):=\sum_{m}R_{mn}[v_m(t)-v_n(t)]^2\Psi\bra{p_n(t)/\peq_n(t),p_m(t)/\peq_m(t)}$.
From Eqs.~\eqref{eq:tmp1} and \eqref{eq:tmp2}, we have
\begin{equation}
\int_0^\tau\avg{2\Dr{\mbm{v}}(t)+\mbm{r}(t),\mbm{q}(t)}dt=0.
\end{equation}
Because $\{\mbm{q}(t)\}_{0\le t\le\tau}$ is an arbitrary perturbation path, the term in the inner product must be zero, i.e., $\Dr{\mbm{v}}(t)+\mbm{r}(t)/2=0$.
Finally, the geodesic equation that determines the shortest path between states $\mbm{p}_0$ and $\mbm{p}_\tau$ is obtained as follows:
\begin{equation}
\begin{cases}
\Dr{\mbm{p}}(t)-\msf{K}_p\mbm{v}(t)=0,\\
\Dr{\mbm{v}}(t)+\displaystyle\frac{1}{2}\mbm{r}(t)=0,\label{eq:Wasserstein.geodesic.equ}
\end{cases}
\end{equation}
with boundary conditions $\mbm{p}(0)=\mbm{p}_0$ and $\mbm{p}(\tau)=\mbm{p}_\tau$.

\subsection{Lower bound of the modified Wasserstein distance in terms of the total variation distance}
Here we derive the lower bound of the Wasserstein distance in terms of the total variation distance, $\msf{d}_{\rm V}(\mbm{p},\mbm{q})=\sum_n|p_n-q_n|$.
In variational form, the distance $\msf{d}_{\rm V}(\mbm{p},\mbm{q})$ can be expressed as
\begin{equation}
\msf{d}_{\rm V}(\mbm{p},\mbm{q})=\max_{\|\mbm{w}\|_\infty\le 1}\,\{\mbm{w}^\top(\mbm{p}-\mbm{q})\}=\max_{\|\mbm{w}\|_\infty\le 1}\,\avg{\mbm{w},\mbm{p}-\mbm{q}},
\end{equation}
where the maximum is taken over all real vectors $\mbm{w}=[w_1,\dots,w_N]^\top$ and $\|\mbm{w}\|_\infty:=\max_n|w_n|$.
Equality is attained when $w_n={\rm sign}(p_n-q_n)$.
Here, the sign function ${\rm sign}(x)$ of $x$ is defined as ${\rm sign}(x)=1$ for $x\ge 0$ and $-1$ otherwise.
By definition of the modified Wasserstein distance, given a fixed positive number $\delta>0$, there exists a smooth curve $\mbm{p}(t)$ with end points $\mbm{p}_0$ and $\mbm{p}_\tau$ such that
\begin{equation}
\tau\int_0^\tau\avg{\mbm{v},\msf{K}_p\mbm{v}}dt\le\mca{W}_{\rm c}(\mbm{p}_0,\mbm{p}_\tau)^2+\delta.
\end{equation}
Here, $\mbm{v}(t)\in\mbb{R}^{N\times 1}$ is determined from $\Dr{\mbm{p}}(t)=\msf{K}_p\mbm{v}(t)$.
For an arbitrary vector $\mbm{w}$ with $\|\mbm{w}\|_\infty\le 1$, we have
\begin{subequations}\label{eq:bound.tot.var.dist}
\begin{align}
\avg{\mbm{w},\mbm{p}_\tau-\mbm{p}_0}&=\int_0^\tau\avg{\mbm{w},\msf{K}_p\mbm{v}}dt\\
&\le\bra{\int_0^\tau\avg{\mbm{w},\msf{K}_p\mbm{w}}dt}^{1/2}\bra{\int_0^\tau\avg{\mbm{v},\msf{K}_p\mbm{v}}dt}^{1/2}\\
&\le\bra{\tau^{-1}\int_0^\tau\avg{\mbm{w},\msf{K}_p\mbm{w}}dt}^{1/2}\bra{\mca{W}_{\rm c}(\mbm{p}_0,\mbm{p}_\tau)^2+\delta}^{1/2}.\label{eq:L1dist.bound.tmp1}
\end{align}
\end{subequations}
To further bound the first term in Eq.~\eqref{eq:L1dist.bound.tmp1}, we apply the inequalities $\Phi(x,y)\le(x+y)/2$ and $(w_n-w_m)^2\le 4$, and obtain
\begin{subequations}
\begin{align}
\avg{\mbm{w},\msf{K}_p\mbm{w}}&=\sum_{m>n}R_{nm}\peq_m\Phi\bra{\frac{p_n}{\peq_n},\frac{p_m}{\peq_m}}\avg{\mbm{w},\msf{E}_{nm}\mbm{w}}\\
&=\sum_{m>n}R_{nm}\peq_m\Phi\bra{\frac{p_n}{\peq_n},\frac{p_m}{\peq_m}}(w_n-w_m)^2\label{eq:class.temp1} \\
&\le 2\sum_{m>n}R_{nm}\peq_m\bra{\frac{p_n}{\peq_n}+\frac{p_m}{\peq_m}}\\
&=2\sum_{m>n}[R_{nm}p_m+R_{mn}p_n].
\end{align}
\end{subequations}
Consequently, we have
\begin{equation}
\mca{W}_{\rm c}(\mbm{p}_0,\mbm{p}_\tau)^2+\delta\ge\frac{\avg{\mbm{w},\mbm{p}_0-\mbm{p}_\tau}^2}{2\tau^{-1}\int_0^\tau\sum_{m>n}[R_{nm}(t)p_m(t)+R_{mn}(t)p_n(t)]dt}.
\end{equation}
Taking the maximum over all $\mbm{w}$ and the limit $\delta\to 0$, we obtain
\begin{equation}
\mca{W}_{\rm c}(\mbm{p}_0,\mbm{p}_\tau)^2\ge\frac{\msf{d}_{\rm V}(\mbm{p}_0,\mbm{p}_\tau)^2}{2\mca{A}_{\rm V}},\label{eq:ClaWasDisLowBou}
\end{equation}
where $\mca{A}_{\rm V}:=\tau^{-1}\int_0^\tau\sum_{m\neq n}R_{mn}(t)\gamma_n(t)dt$ is the average dynamical activity along the geodesic path $\{\mbm{\gamma}(t)\}_{0\le t\le \tau}$.
The dynamical activity characterizes the time scale of the system.
As it indicates the time-symmetric changes in the system, it plays important roles in nonequilibrium phenomena \cite{Maes.2020.PR}.
From Eqs.~\WassBoundCla~and \eqref{eq:ClaWasDisLowBou}, the classical speed limits of the state transformation are obtained as
\begin{equation}
\tau\ge\frac{\mca{W}_{\rm c}(\mbm{p}(0),\mbm{p}(\tau))^2}{\Delta S_{\rm tot}}\ge\frac{\msf{d}_{\rm V}(\mbm{p}(0),\mbm{p}(\tau))^2}{2\Delta S_{\rm tot}\mca{A}_{\rm V}}.\label{eq:ClaSpeLim}
\end{equation}
These inequalities imply a trade-off relation between the time needed to transform the system state and the physical quantities such as entropy production and dynamical activity.
Specifically, a fast transformation necessitates high dissipation and frenesy.
The last bound in the inequality \eqref{eq:ClaSpeLim} is analogous to, but distinct from, a bound derived in Ref.~\cite{Shiraishi.2018.PRL}.
In the earlier study, $\mca{A}_{\rm V}$ is replaced by the average dynamical activity along the path described by the time evolution of the system.

We can derive a bound that is tighter than Eq.~\eqref{eq:ClaWasDisLowBou}.
First, we divide the set of states $\mca{N}=\{1,\dots,N\}$ into two subsets $\mca{N}=\mca{X}_-\cup\mca{X}_+$, where
\begin{align}
\mca{X}_-&:=\{n~|~1\le n\le N,~ p_n(0)<p_n(\tau)\},\\
\mca{X}_+&:=\{n~|~1\le n\le N,~ p_n(0)\ge p_n(\tau)\}.
\end{align}
For convenience, we define ${\rm par}(n):=-1$ if $n\in\mca{X}_-$ and ${\rm par}(n):=1$ if $n\in\mca{X}_+$.
Then, $\msf{d}_{\rm V}(\mbm{p}_0,\mbm{p}_\tau)=\avg{\mbm{w},\mbm{p}_0-\mbm{p}_\tau}$, where $w_n={\rm par}(n)$.
From Eq.~\eqref{eq:class.temp1}, noticing that $(w_n-w_m)^2=0$ if $n$ and $m$ belong to the same subset $\mca{X}_-$ or $\mca{X}_+$, and $(w_n-w_m)^2=4$ if $n$ and $m$ belong to different subsets, we have
\begin{subequations}
\begin{align}
\avg{\mbm{w},\msf{K}_p\mbm{w}}&\le 2\sum_{m>n~{\&}~{\rm par}(m)\neq{\rm par}(n)}R_{nm}\peq_m\bra{\frac{p_n}{\peq_n}+\frac{p_m}{\peq_m}}\\
&=2\sum_{m\in\mca{X}_-,\,n\in\mca{X_+}}[R_{mn}p_n+R_{nm}p_m]=2\sum_{{\rm par}(m)\neq{\rm par}(n)}R_{mn}p_n.
\end{align}
\end{subequations}
Subsequently, we obtain a tighter bound as
\begin{equation}
\mca{W}_{\rm c}(\mbm{p}_0,\mbm{p}_\tau)^2\ge\frac{\msf{d}_{\rm V}(\mbm{p}_0,\mbm{p}_\tau)^2}{2\mca{A}_{\rm V}^{\rm par}},
\end{equation}
where $\mca{A}_{\rm V}^{\rm par}:=\tau^{-1}\int_0^\tau\sum_{{\rm par}(m)\neq{\rm par}(n)}R_{mn}(t)\gamma_n(t)dt$ is the average of the {\it partial} dynamical activity along the geodesic path $\{\mbm{\gamma}(t)\}_{0\le t\le \tau}$.
Since $\mca{A}_{\rm V}^{\rm par}$ involves only transition rates between states in $\mca{X}_-$ and $\mca{X}_+$, $\mca{A}_{\rm V}^{\rm par}\le \mca{A}_{\rm V}$.

Following the same approach (i.e., applying Eq.~\eqref{eq:bound.tot.var.dist} for the path described by the system dynamics), we obtain
\begin{subequations}
\begin{align}
\avg{\mbm{w},\mbm{p}_\tau-\mbm{p}_0}&=\int_0^\tau\avg{\mbm{w},\msf{K}_p\mbm{h}}dt\\
&\le\bra{\int_0^\tau\avg{\mbm{w},\msf{K}_p\mbm{w}}dt}^{1/2}\bra{\int_0^\tau\avg{\mbm{h},\msf{K}_p\mbm{h}}dt}^{1/2}\\
&=\bra{\int_0^\tau\avg{\mbm{w},\msf{K}_p\mbm{w}}dt}^{1/2}\sqrt{\Delta S_{\rm tot}}.
\end{align}
\end{subequations}
Analogously, we can prove that
\begin{equation}
\tau\ge\frac{\msf{d}_{\rm V}(\mbm{p}(0),\mbm{p}(\tau))^2}{2\Delta S_{\rm tot}\mca{A}^{\rm par}},\label{eq:tighter.csl}
\end{equation}
where $\mca{A}^{\rm par}:=\tau^{-1}\int_0^\tau\sum_{{\rm par}(m)\neq{\rm par}(n)}R_{mn}(t)p_n(t)dt$ is the average of the {\it partial} dynamical activity along the path described by the time evolution of the system.
Obviously, Eq.~\eqref{eq:tighter.csl} is tighter than the bound reported in Ref.~\cite{Shiraishi.2018.PRL}.
The newly obtained bound indicates that the speed of the state transformation is not constrained by the total dynamical activity, but by the partial dynamical activity induced by transitions between states in $\mca{X}_-$ and $\mca{X}_+$.

\subsection{Thermalization process of a three-level system}
Here we illustrate the derived bound on the thermalization process of a three-level system.
The transition rates are time-independent and equal to
\begin{equation}
R_{mn}=w_{mn}e^{\beta(\mca{E}_n-\mca{E}_m)/2}{\rm sech}[\beta(\mca{E}_n-\mca{E}_m)/2],
\end{equation}
where $w_{mn}=w_{nm}$ are nonnegative constants.
Evidently, the transition rates satisfy the detailed balance conditions $R_{mn}\peq_n=R_{nm}\peq_m$.
According to Eq.~\WassBoundCla~in the main text, the entropy production is bounded from below by the modified Wasserstein distance as
\begin{equation}
\Delta S_{\rm tot}\ge \frac{\mca{W}_{\rm c}(\mbm{p}(0),\mbm{p}(\tau))^2}{\tau}.
\end{equation}
The total entropy production can be explicitly expressed as $\Delta S_{\rm tot}=D(\mbm{p}(0)||\mbm{p}^{\rm eq})-D(\mbm{p}(\tau)||\mbm{p}^{\rm eq})$.
In thermalization processes satisfying the detailed balance conditions, Ref.~\cite{Shiraishi.2019.PRL} proved that the relative entropy satisfies the reverse triangle inequality:
\begin{equation}
D(\mbm{p}(0)||\mbm{p}^{\rm eq})\ge D(\mbm{p}(0)||\mbm{p}(\tau))+D(\mbm{p}(\tau)||\mbm{p}^{\rm eq}).
\end{equation}
Subsequently, the entropy production during thermalization processes is bounded from below by an information-theoretical quantity of the initial and final states, $\Delta S_{\rm tot}\ge D(\mbm{p}(0)||\mbm{p}(\tau))$.
For fixed transition rates, Fig.~\ref{fig:three.level} plots the entropy production, modified Wasserstein distance, and relative entropy as functions of time $\tau$.
In this figure, the distance term $\mca{W}_{\rm c}^2(\mbm{p}(0),\mbm{p}(\tau))/\tau$ and the relative entropy always lie below the entropy production $\Delta S_{\rm tot}$.
The modified Wasserstein distance is tight in the short-time regime, whereas the relative entropy saturates in the long-time limit.
Therefore, these two bounds complementarily characterize the irreversibility in thermalization processes.
\begin{figure}
\centering
	\includegraphics[width=0.8\linewidth]{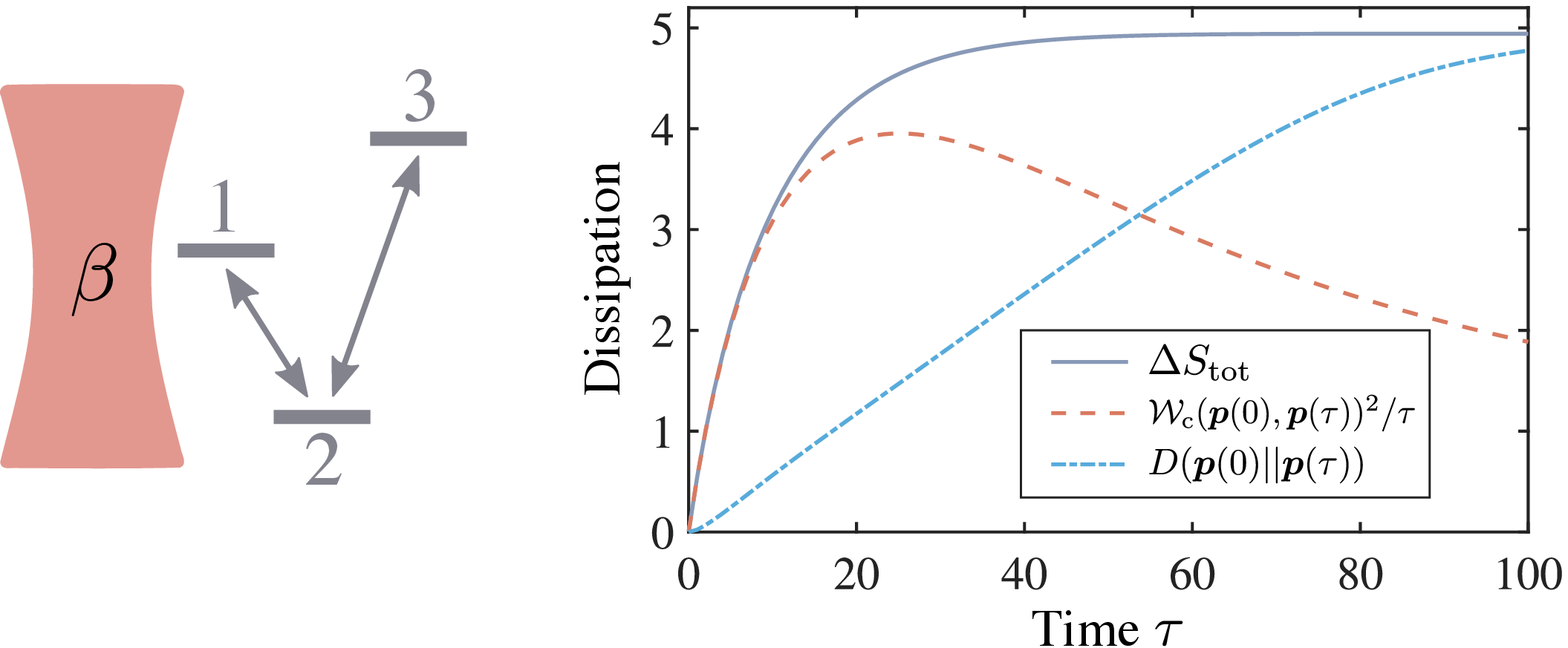}
	\protect\caption{Numerical verification of the derived bound. $\Delta S_{\rm tot}$ (solid line), $\mca{W}_{\rm c}(\mbm{p}(0),\mbm{p}(\tau))^2/\tau$ (dashed line), and $D(\mbm{p}(0)||\mbm{p}(\tau))$ (dash-dotted line) during the thermalization process of a three-level system. Parameters are set as $\beta=1,w_{12}=1,w_{23}=2,w_{13}=0,\mca{E}_1=3,\mca{E}_2=-0.5,\mca{E}_3=6$, and $\mbm{p}(0)=[0.1,0.1,0.8]^\top$.}\label{fig:three.level}
\end{figure}